\documentclass[12pt,a4paper]{article}

\usepackage{amssymb}
\usepackage{amsmath}
\usepackage{amsthm}

\theoremstyle{plain}
\newtheorem{theorem}{Theorem}[section]
\newtheorem{proposition}[theorem]{Proposition}
\newtheorem{lemma}[theorem]{Lemma}
\newtheorem{corollary}[theorem]{Corollary}

\theoremstyle{definition}
\newtheorem{definition}[theorem]{Definition}

\newtheorem{assumption}[theorem]{Assumption}

\theoremstyle{remark}
\newtheorem{remark}[theorem]{Remark}
\newtheorem*{remark*}{Remark}

\numberwithin{equation}{section}

\DeclareMathOperator{\Ker}{Ker}
\newcommand{\ket}[1]{\lvert{#1}\rangle}
\newcommand{\bra}[1]{\langle{#1}\rvert} 
\newcommand{\ip}[2]{\langle{#1},{#2}\rangle}
\begin{document}
\allowdisplaybreaks
\title{Resolvent expansions for the Schr\"{o}dinger
operator on the discrete half-line
}
\author{Kenichi {\scshape Ito}\footnote{Department of Mathematics, Graduate School of Science, Kobe University,
1-1, Rokkodai, Nada-ku, Kobe 657-8501, Japan.
E-mail: \texttt{ito-ken@math.kobe-u.ac.jp}.}
\and
Arne {\scshape Jensen}\footnote{Department of Mathematical Sciences,
Aalborg University, Fredrik Bajers Vej 7G, DK-9220 Aalborg \O{}, Denmark.
E-mail: \texttt{matarne@math.aau.dk}.}
}

\date{}
\maketitle

%\title{Resolvent expansions for the Schr\"{o}dinger
%operator on the discrete half-line} 
%
%\author{Kenichi Ito}
%\email{ito-ken@math.kobe-u.ac.jp}
%\affiliation{Department of Mathematics, Graduate School of Science, Kobe University,
%1-1, Rokkodai, Nada-ku, Kobe 657-8501, Japan.}
%
%\author{Arne Jensen}
%\email{matarne@math.aau.dk}
%\affiliation{Department of Mathematical Sciences,
%Aalborg University, Fredrik Bajers Vej 7G, DK-9220 Aalborg \O{}, Denmark.}
%
%\date{\today}

\begin{abstract}
Simplified models of transport in mesoscopic systems are often based on a small sample connected to a finite number of leads. The leads are often modelled using the Laplacian on the discrete half-line $\mathbb N$. Detailed studies of the transport  near thresholds require detailed information on the resolvent of the Laplacian on the discrete half-line.
This paper presents a complete study of threshold
resonance states and resolvent expansions at a threshold 
for the Schr\"odinger operator  
on the discrete half-line $\mathbb N$ 
with a general boundary condition. 
A precise description of the expansion coefficients
reveals their exact correspondence to 
the generalized eigenspaces, or the threshold types.
The presentation of the paper is adapted from 
that of Ito-Jensen [Rev.\ Math.\ Phys.\ {\bf 27} (2015), 1550002 (45 pages)],
implementing the expansion scheme of 
Jensen-Nenciu [Rev.\ Math.\ Phys.\ \textbf{13} (2001), 717--754, 
\textbf{16} (2004), 675--677]
in its full generality.
\end{abstract}

%\pacs{}% insert suggested PACS numbers in braces on next line

\section{Introduction}
Simplified models of transport in mesoscopic systems are often based on a small sample connected to a finite number of leads. The leads are often modelled using the Laplacian on the discrete half-line $\mathbb N$. Detailed studies of the transport near thresholds require detailed information on the resolvent of the Laplacian on the discrete half-line. For an example see Cornean-Jensen-Nenciu\cite{CJN} and references therein. The results in this paper allow one to obtain more detailed information on the adiabatic limit studied in Cornean-Jensen-Nenciu\cite{CJN}.

Let $H_0$ be the positive Laplacian
on the discrete half-line $\mathbb N=\{1,2,\ldots\}$, i.e., 
for any sequence $x\colon \mathbb N\to \mathbb C$ we define the sequence
$H_0x\colon \mathbb N\to \mathbb C$ by
\begin{equation}
(H_0x)[n]=-(x[n+1]+x[n-1]-2x[n]).
\label{160814}
\end{equation}
The definition \eqref{160814} is incomplete without 
assigning a \textit{boundary condition}, 
or a \textit{boundary value} $x[0]$ 
for each sequence $x\colon\mathbb N\to\mathbb C$. 
In this paper we focus on 
the \textit{Dirichlet boundary condition}:
\begin{align}
x[0]=0.
\label{160820}
\end{align}
In other words, we set for any sequence $x\colon \mathbb N\to \mathbb C$
\begin{equation}\label{D-formula}
(H_0x)[n]=\begin{cases}
2x[1]-x[2]& \text{for }n=1,\\
2x[n]-x[n+1]-x[n-1] & \text{for }n\geq2.
\end{cases}
\end{equation}

The restriction of $H_0$ to the Hilbert space $\mathcal H=\ell^2(\mathbb{N})$
is bounded and self-adjoint,
and its spectrum is 
\begin{align}
\sigma(H_0)=\sigma_{\mathrm{ac}}(H_0)=[0,4].
\label{1608202240}
\end{align}
The points $0,4\in \sigma(H_0)$ are called the \textit{thresholds}.
The purpose of this paper is to analyze the threshold behavior 
of a perturbed Laplacian $H=H_0+V$ on the discrete half-line $\mathbb N$. 
We compute an asymptotic expansion of the resolvent $R(z)=(H-z)^{-1}$
at the threshold $z=0$, and, in particular, 
describe a precise relation between the expansion 
coefficients and the generalized eigenspaces.
The generalized eigenspace considered here 
is the largest possible one, and includes the threshold resonance states as a part of it. 
These investigations are done in the same manner as in Ito-Jensen\cite{IJ},
employing the expansion scheme given in Jensen-Nenciu\cite{JN0,JN1}. The technique used in Ito-Jensen\cite{IJ} to treat the threshold $4$ can be applied here. Hence we discuss only the threshold zero.

The starting point of our analysis is the free resolvent kernel discussed 
in Section~\ref{16082022}. 
The main results of the paper will be presented in Section~\ref{1608202219}.
Actually general boundary conditions are included in our setting 
as specific forms of perturbations of the Dirichlet Laplacian. 
We will see this in Section~\ref{16081823}. Section~\ref{12.12.19.2.5} is devoted to an analysis of the generalized eigenspace.
After a short preliminary presentation in Section~\ref{1608217},
the proofs of the main theorems will be provided in 
Sections~\ref{1608241726}--\ref{1608241729}
according to each threshold type. 
There we will repeatedly use the inversion formula from Jensen-Nenciu\cite{JN0}, 
adapted to the case at hand.
As a reference we will quote the formula in the form given in Ito-Jensen\cite{IJ} 
in Appendix~\ref{1608211941}.

There is a large number of papers on discrete Schr\"{o}dinger operators. However, as far as we are aware, the complete threshold analyses and the resolvent expansions presented here are new.

\section{The free Laplacian}\label{16082022}

In this section we discuss properties of 
the free Dirichlet Laplacian $H_0$  
on the discrete half-line $\mathbb N$ defined by 
\eqref{160814} and \eqref{160820}, or by \eqref{D-formula}.
The properties presented here 
may be considered as a prototype of our main results for a perturbed
Laplacian. They will be employed repeatedly 
both in stating and in proving the main theorems. 

Let $\widehat{\mathcal H}=L^2(0,\pi)$, 
and define the Fourier transform 
$\mathcal F\colon\mathcal H\to \widehat{\mathcal H}$
and its inverse 
$\mathcal F^*\colon\widehat{\mathcal H}\to \mathcal H$
by
\begin{align*}
(\mathcal Fx)(\theta)
&=\sqrt{2/\pi}\sum_{n=1}^\infty x[n]\sin(n\theta)
,\\
(\mathcal F^*f)[n]
&=\sqrt{2/\pi}\int_0^\pi f(\theta)\sin(n\theta)\,d\theta.
\end{align*}
Then we have a spectral representation of $H_0$: 
\begin{align}
\mathcal FH_0\mathcal F^*
=2-2\cos\theta=4\sin^2(\theta/2).
\label{1608146}
\end{align}
This in fact verifies \eqref{1608202240}.
Using the expression \eqref{1608146},
or antisymmetrizing the kernel of resolvent on the whole line $\mathbb Z$,
see e.g.\ Ito-Jensen\cite{IJ}, 
we can compute the kernel of resolvent $R_0(z)=(H_0-z)^{-1}$:
For $z\in \mathbb C\setminus [0,4]$ with $z\sim 0$ we have
\begin{align}\label{D-resolvent-kernel}
R_0(z)[n,m]=\frac{i}{2\sin \phi}\bigl(e^{i\phi|n-m|}-e^{i\phi(n+m)}\bigr),
\quad n,m\in\mathbb N.
\end{align}
Here the variable $z\in \mathbb{C}\setminus [0,4]$ is related to $\phi$
through the correspondence 
\begin{align*}
z=4\sin^2(\phi/2),\quad \mathop{\mathrm{Im}}\phi>0.
%\label{11.1.26.15.17}
\end{align*}
Using the expression \eqref{D-resolvent-kernel}, we can explicitly
compute the expansion of $R_0(z)$ around $z=0$.
Before stating it let us introduce the notation employed in this paper.

\medskip
\noindent
{\bf Notation.}
In expansions we change variable from $z\in\mathbb{C}\setminus [0,\infty)$ to 
$\kappa$. These variables are related as 
\begin{equation}
\kappa=-i\sqrt{z}, \quad \mathop{\mathrm{Im}} z >0,\quad \mathop{\mathrm{Im}}\sqrt{z}>0.
\end{equation}
We freely write $R(z)$ as $R(\kappa)$, etc. We use the notation 
$$n\wedge m=\min\{n,m\},\quad n\vee m=\max\{n,m\}.$$
For $s\in \mathbb R$ we let 
\begin{align*}
\mathcal L^s&=
\ell^{1,s}(\mathbb{N})\\
&=\bigl\{x\colon \mathbb{N}\to \mathbb{C};\\ 
&\qquad\|x\|_{1,s}=\sum_{n\in\mathbb{N}}(1+n^2)^{s/2}|x[n]|<\infty\bigr\},\\
(\mathcal L^s)^*
&=\ell^{\infty,-s}(\mathbb{N})\\
&=\bigl\{x\colon \mathbb{N}\to \mathbb{C};\\
&\qquad\|x\|_{\infty,-s}=\sup_{n\in\mathbb{N}}(1+n^2)^{-s/2}|x[n]|<\infty\bigr\}.
\end{align*}
We denote the set of all bounded operators 
from a general Banach space $\mathcal K$ to another 
$\mathcal K'$ by $\mathcal B(\mathcal K,\mathcal K')$,
and abbreviate $\mathcal B(\mathcal K)=\mathcal B(\mathcal K,\mathcal K)$.
In particular, we write
$${\mathcal B}^s=\mathcal B(\mathcal L^s,(\mathcal L^s)^*).$$
We replace $\mathcal B$ by $\mathcal C$ when considering 
the corresponding spaces of compact operators.
Define the sequences 
$\mathbf n\in(\mathcal L^1)^*$ and $\mathbf 1\in (\mathcal L^0)^*$ by
\begin{equation}\label{def-1-n}
\mathbf{n}[m]=m\quad\text{and}\quad\mathbf{1}[m]=1,\quad m\in\mathbb N,
\end{equation}
respectively. 
Throughout the paper we frequently use  the 
\textit{pseudo-inverse} $A^\dagger$
of a self-adjoint operator $A$. For this concept we refer to Appendix~\ref{1608211941}.

\begin{proposition}\label{prop12}
Let $N\ge 0$ be any integer.
As $\kappa\to 0$ with $\mathop{\mathrm{Re}}\kappa>0$,
the resolvent $R_0(\kappa)$ has the expansion:
\begin{align}
R_0(\kappa)=\sum_{j=0}^N\kappa^jG_{0,j}+{\mathcal O}(\kappa^{N+1})
\quad \text{in }\mathcal B^{N+2},
\label{free-expan}
\end{align}
with $G_{0,j}\in\mathcal B^{j+1}$ for $j$ even,
and $G_{0,j}\in\mathcal B^{j}$ for $j$ odd, satisfying
\begin{equation}
\begin{split}
H_0G_{0,0}&=G_{0,0}H_0=I,\\ 
H_0G_{0,1}&=G_{0,1}H_0=0,\\ 
H_0G_{0,j}&=G_{0,j}H_0=-G_{0,j-2}\quad \text{for }j\ge 2.
\end{split}
\label{18082023}
\end{equation}
The coefficients $G_{0,j}$ have explicit kernels,
and the first few are given by 
\begin{align}
G_{0,0}[n,m]&=n\wedge m,
\label{G00}\\
G_{0,1}[n,m]&=-n\cdot m,\label{G01}\\
G_{0,2}[n,m]
&=-\tfrac{1}{6}(n\wedge m)\notag\\
&\quad+\tfrac{1}{6}(n\wedge m)^3+\tfrac{1}{2}n\cdot m\cdot(n\vee m),
\label{G02}\\
\begin{split}
G_{0,3}[n,m]
&=\tfrac{5}{24}n\cdot m-\tfrac{1}{6}n^3\cdot m-\tfrac{1}{6}n\cdot m^3.
\end{split}
\label{G03}
\end{align}
\end{proposition}
\begin{proof}
The expansion \eqref{free-expan} with expressions \eqref{G00}--\eqref{G03}
follows directly from \eqref{D-resolvent-kernel},
cf.\ Ito-Jensen\cite[Proposition 2.1]{IJ}.
To see the identities in \eqref{18082023} it suffices to 
note that for any rapidly decreasing sequence $\Psi\colon \mathbb N\to\mathbb C$
we have 
\begin{align*}
(H_0+\kappa^2)R_0(\kappa)\Psi=R_0(\kappa)(H_0+\kappa^2)\Psi=\Psi
\end{align*}
for $\mathop{\mathrm{Re}}\kappa>0$. The details of the computations are omitted.
\end{proof}

We note that the sequence $\mathbf n\in (\mathcal L^1)^*$ 
is a \textit{generalized eigenfunction} for $H_0$, 
and the coefficient $G_{0,1}$ is a 
\textit{generalized projection} onto it:
\begin{align*}
H_0\mathbf n=0,\quad 
G_{0,1}=-|\mathbf n\rangle\langle\mathbf n|.
\end{align*}
On the other hand, the sequence $\mathbf 1\in (\mathcal L^0)^*$,
which with $\mathbf n$ forms a basis of 
the generalized eigenspace for the Laplacian on the whole line $\mathbb Z$,
is not a generalized eigenfunction on $\mathbb N$.
It does not appear in the above expansion coefficients, either.

\section{The perturbed Laplacian}\label{1608202219}

Now we consider the perturbed Laplacian $H=H_0+V$ on $\mathbb N$,
and state the main theorems of the paper.
These theorems reveal a precise relation between the 
\textit{generalized eigenspace} and the expansion coefficients of the resolvent
at threshold. 

The class of interactions considered here is from Ito-Jensen\cite{IJ}.
It is general enough to contain non-local interactions,
but is formulated a little abstractly.
We refer to Ito-Jensen\cite[Appendix B]{IJ} for examples. We note that this class of interactions is closed under addition, see Ito-Jensen\cite{IJ}.

Recall the notation defined right before Proposition~\ref{prop12}.

\begin{assumption}\label{assumV}
Let $V\in {\mathcal B}(\mathcal H)$ be self-adjoint,
and assume that there exist 
an injective operator $v\in \mathcal B({\mathcal K},\mathcal L^\beta)
\cap \mathcal C({\mathcal K},\mathcal L^1)$
with $\beta\ge 1$
and a self-adjoint unitary operator
$U\in \mathcal B({\mathcal K})$,
both defined on some Hilbert space ${\mathcal K}$,
such that 
\begin{equation*}
V=vUv^*\in \mathcal B((\mathcal L^{\beta})^*,\mathcal L^\beta)
\cap 
\mathcal C((\mathcal L^1)^*,\mathcal L^1).
\end{equation*}
\end{assumption}

Under Assumption~\ref{assumV}
we let 
$$H=H_0+V,\quad R(z)=(H-z)^{-1}.$$ 
The operator $H$ is a bounded self-adjoint operator on $\mathcal H$ with $\sigma_{\rm ess}(H)=[0,4]$. Using the Mourre method (see Boutet de Monvel-Shabani\cite{BMS}) one can show that $\sigma_{\rm sc}(H)=\emptyset$. For local $V$ other conditions for 
$\sigma_{\rm sc}(H)=\emptyset$ are given in Damanik-Killip\cite{DK}.

Let us consider the solutions to the equation $H\Psi=0$ in the largest space where it can be defined.  Define the (\textit{generalized}) \textit{zero eigenspaces} by
\begin{align}
\widetilde{\mathcal{E}}&=\{\Psi\in(\mathcal{L}^{\beta})^*|\, H\Psi=0\},\\
\mathcal{E}&=\widetilde{\mathcal{E}}\cap(\mathbb{C}\mathbf{1}\oplus\mathcal{L}^{\beta-2}),\label{16071616}\\
\mathsf{E}&=\widetilde{\mathcal{E}}\cap\mathcal{L}^{\beta-2}.\label{16071617}
\end{align}
These spaces will be analyzed in detail in Section~\ref{12.12.19.2.5}.
Here we only quote some of the results given there: 
Under Assumption~\ref{assumV} with $\beta\ge 1$
the generalized eigenfunctions have a specific asymptotics:
\begin{align}
\widetilde{\mathcal{E}}\subset\mathbb{C}\mathbf{n}\oplus\mathbb{C}\mathbf{1}\oplus\mathcal{L}^{\beta-2},
\label{13.3.7.13.48}
\end{align}
and their dimensions satisfy 
\begin{align*}
\dim(\widetilde{\mathcal E}/\mathcal E)+\dim(\mathcal E/\mathsf E)=1,\quad 
0\le \dim\mathsf E<\infty.
\end{align*}
We introduce the same classification of the threshold  
as in Ito-Jensen\cite[Definition 1.6]{IJ}.

\begin{definition}\label{def-reg-excp}
The threshold $z=0$ is said to be 
\begin{enumerate}
\item 
a \textit{regular point}, 
if $\mathcal{E}= \mathsf{E}= \{0\}$;
\item
an \textit{exceptional point of the first kind}, 
if $\mathcal{E}\supsetneq \mathsf{E}= \{0\}$;
\item
an \textit{exceptional point of the second kind}, 
if 
$\mathcal{E}=\mathsf{E}\supsetneq \{0\}$;
\item
an \textit{exceptional point of the third kind}, 
if 
$\mathcal{E}\supsetneq \mathsf{E}\supsetneq \{0\}$.
\end{enumerate}
\end{definition}

It would be more precise to call a
function in $\widetilde{\mathcal E}$ a \textit{generalized eigenfunction},
that in $\mathcal E$ a \textit{resonance function},
and that in $\mathsf E$ an \textit{eigenfunction},
but sometimes all of them are called simply \textit{eigenfunctions}.
In particular, we call $\Psi_c\in \mathcal E$ a
\textit{canonical resonance function} if it satisfies
\begin{align*}
\forall \Psi\in \mathsf E\ \ 
\langle \Psi, \Psi_c\rangle=0,
\quad \text{and}\quad 
\Psi_c- \mathbf 1\in\mathcal L^{\beta-2}.
\end{align*}
We remark that the latter asymptotics for $\Psi_c\in \mathcal E$
is equivalent to 
\begin{align*}
\langle V\mathbf n,\Psi_c\rangle=-1.
\end{align*}
We will prove this equivalence in Proposition~\ref{13.1.16.2.51}.

We now state the resolvent expansions in the four cases given 
in Definition~\ref{def-reg-excp}. We impose assumptions 
on the parameter $\beta$ from Assumption~\ref{assumV} in each of the four cases. For simplicity we state the results for integer values of $\beta$. The extension to general $\beta$ is straightforward but leads to more complicated statements of the results and requires a different approach to the error estimates in the theorems below.
Let us set 
\begin{align*}
M_0=U+v^*G_{0,0}v\colon\mathcal K\to\mathcal K,
\end{align*}
and denote its pseudo-inverse by $M_0^\dagger$,
see Appendix~\ref{1608211941}.

\begin{theorem}\label{thm-reg}
Assume that the threshold $0$ is a \emph{regular} point,
and that Assumption~\ref{assumV} is fulfilled for some integer $\beta\ge 2$.
Then
\begin{equation}
R(\kappa)=\sum_{j=0}^{\beta-2}\kappa^jG_j+\mathcal{O}(\kappa^{\beta-1})
\quad
\text{in }\mathcal{B}^{\beta-2}
\end{equation}
with 
$G_j\in\mathcal B^{j+1}$ for $j$ even,
and $G_j\in\mathcal B^{j}$ for $j$ odd.
The coefficients $G_j$ can be computed explicitly.
The first two coefficients can be expressed as  
\begin{align}
G_0&
=G_{0,0}-G_{0,0} vM_0^\dagger v^*G_{0,0},\label{160822427}
\\
G_1&=-|\widetilde\Psi_c\rangle\langle\widetilde\Psi_c|,\label{160822428}
\end{align}
where $\widetilde\Psi_c\in\widetilde{\mathcal E}$ is 
a generalized eigenfunction with asymptotics
\begin{align*}
m^{-1}\widetilde\Psi_c[m]\to 1\quad \text{as }m\to\infty.
\end{align*}
\end{theorem}
\begin{remark}
Under the assumption of Theorem~\ref{thm-reg}
the operator $M_0$ is actually invertible: $M_0^\dagger =M_0^{-1}$.
The operators $I+G_{0,0}V$ and $I+VG_{0,0}$ are also invertible,
and we have the expressions
\begin{align}
I-G_{0,0} vM_0^\dagger v^*
&=(I+G_{0,0}V)^{-1},\\
I- vM_0^\dagger v^*G_{0,0}
&=(I+VG_{0,0})^{-1}.
\label{16082620}
\end{align}
We will verify these right after the proof of Theorem~\ref{thm-reg}.
\end{remark}

\begin{theorem}\label{thm-ex1}
Assume that the threshold $0$ is an exceptional point of the \emph{first kind},
and that Assumption~\ref{assumV} is fulfilled for some integer $\beta\ge 3$.
Then
\begin{equation}\label{expand-first}
R(\kappa)=\sum_{j=-1}^{\beta-4}\kappa^jG_j+\mathcal{O}(\kappa^{\beta-3})\quad
\text{in }\mathcal{B}^{\beta-1}
\end{equation}
with 
$G_j\in\mathcal B^{j+3}$ for $j$ even,
and $G_j\in\mathcal B^{j+2}$ for $j$ odd.
The coefficients $G_j$ can be computed explicitly.
The first two coefficients can be expressed as 
\begin{align}
G_{-1}&=\ket{\Psi_c}\bra{\Psi_c},
\label{1608229}
\\
\begin{split}
G_0&
=G_{0,0}
-\bigl(G_{0,0}-|\Psi_c\rangle\langle\mathbf n|\bigr)
vM_0^\dagger v^*\bigl(G_{0,0}
-\bigl|\mathbf n\bigr\rangle\langle \Psi_c| \bigr)
\\&\quad
-\bigl[\| \Psi_c-\mathbf 1\|^2+2\mathop{\mathrm{Re}}\langle\mathbf 1,\Psi_c-\mathbf 1\rangle-\tfrac12\big]|\Psi_c\rangle\langle \Psi_c| 
\\&\quad
- |\Psi_c\rangle\bigl\langle\mathbf n\bigr|
-\bigl|\mathbf n\bigr\rangle\langle \Psi_c| 
,
\end{split}
\label{16082513}
\end{align}
where $\Psi_c\in\mathcal E$ is the canonical resonance function.
\end{theorem}

\begin{theorem}\label{thm-ex2}
Assume that the threshold $0$ is an exceptional point of the \emph{second kind},
and that Assumption~\ref{assumV} is fulfilled for some integer $\beta\ge 4$.
Then
\begin{equation}\label{expand-second}
R(\kappa)=\sum_{j=-2}^{\beta-6}\kappa^jG_j+\mathcal{O}(\kappa^{\beta-5})
\quad\text{in }
\mathcal B^{\beta-2}
\end{equation}
with $G_j\in\mathcal B^{j+3}$ for $j$ even,
and $G_j\in\mathcal B^{j+2}$ for $j$ odd.
The coefficients $G_j$ can be computed explicitly.
The first four coefficients can be expressed as 
\begin{align}
G_{-2}&=P_0,\label{ex2-G-2}
\\
G_{-1}&=0,\label{ex2-G-1}
\\
G_{0}&=
(I-P_0)
\bigl(G_{0,0}
-G_{0,0}
vM_0^\dagger v^*G_{0,0}
\bigr)(1-P_0)
,\label{ex2-G0}
\\
G_1&=
(I-P_0)
\bigl(I-G_{0,0}vM_0^\dagger v^*\bigr)G_{0,1}\notag\\
&\quad
\times\bigl(I-vM_0^\dagger v^*G_{0,0}\bigr)
(I-P_0)\notag
\\&\quad
-P_0G_{0,0}vM_0^\dagger  v^*G_{0,1}vM_0^\dagger v^*G_{0,0}P_0,
\label{ex2-G1}
\end{align}
where $P_0$ is the projection onto $\mathsf E$.
\end{theorem}

\begin{theorem}\label{thm-ex3}
Assume that the threshold $0$ is an exceptional point of the \emph{third
kind},
and that Assumption~\ref{assumV} is fulfilled for some integer $\beta\ge 4$.
Then
\begin{equation}\label{expand-third}
R(\kappa)=\sum_{j=-2}^{\beta-6}\kappa^jG_j+\mathcal{O}(\kappa^{\beta-5})
\quad
\text{in }\mathcal B^{\beta-2}
\end{equation}
with $G_j\in\mathcal B^{j+3}$ for $j$ even,
and $G_j\in\mathcal B^{j+2}$ for $j$ odd.
The coefficients $G_j$ can be computed explicitly.
The first two coefficients can be expressed as 
\begin{align*}
G_{-2}&=P_0,%\label{G-2}
\\
G_{-1}&=\ket{\Psi_c}\bra{\Psi_c}%\label{G-1}
,
\end{align*}
where $P_0$ is the projection onto $\mathsf E$,
and $\Psi_c\in\mathcal E$
is the canonical resonance function.
\end{theorem}

By Theorems~\ref{thm-reg}--\ref{thm-ex3}, if $\beta\ge 4$, 
the resolvent $R(\kappa)$ always has an expansion of some order,
and its threshold type can be determined by the coefficients $G_{-2}$ and $G_{-1}$.
We also state as a corollary certain identities satisfied by 
the coefficients. 

\begin{corollary}
The coefficients $G_j$ 
from Theorems~\ref{thm-reg}--\ref{thm-ex3} 
satisfy
\begin{align*}
HG_j&=G_jH=0\quad\text{for }j=-2,-1,\\
HG_0&=G_0H=I-P_0,\\
HG_j&=G_jH=-G_{j-2}\quad\text{for }j\ge 1,
\end{align*}
where $P_0$ is the projection onto $\mathsf E$.
\end{corollary}
\begin{proof}
The assertion is verified by Theorems~\ref{thm-reg}--\ref{thm-ex3} and 
the identities
\begin{align*}
(H+\kappa^2)R(\kappa)\Psi=R(\kappa)(H+\kappa^2)\Psi=\Psi\quad 
\end{align*}
for any rapidly decreasing function $\Psi\colon\mathbb N\to\mathbb C$ and 
any $\kappa\sim 0$ with $\mathop{\mathrm{Re}}\kappa>0$.
\end{proof}

We shall prove Theorems~\ref{thm-reg}--\ref{thm-ex3} 
following the procedure given in Ito-Jensen\cite{IJ}. 
The proofs will be given 
in Sections~\ref{1608241726}--\ref{1608241729}
with preliminaries in the preceding sections.

\section{General boundary conditions}\label{16081823}

In this section we comment on discrete analogues 
of general boundary conditions at the origin of the half-line,
such as the Neumann and the Robin conditions. 
In particular, we introduce specific potentials that allows us to 
deal with such a general boundary condition as 
a perturbation of the Dirichlet condition.

On the discrete half-line 
a boundary condition is realized simply 
by assigning a value to $x[0]$ for each function $x\colon \mathbb N\to\mathbb C$,
as in \eqref{160820}. 
The natural realization of the Neumann boundary condition is 
to assign the difference there to be $0$, i.e.,
\begin{align*}
x[1]-x[0]=0\quad \text{or}\quad x[0]=x[1].
%\label{1608214}
\end{align*}
Similarly, a more general Robin condition is realized by setting
\begin{align*}
ax[0]+b(x[0]-x[1])=0;\quad (a,b)\neq (0,0).
\end{align*}
Here we may take $a\neq -b$. Otherwise 
it reduces to the \textit{shifted} Dirichlet condition $x[1]=0$.

Let us remark that there is yet another realization 
of the Dirichlet boundary condition:
\begin{align}
x[0]=-x[1], \label{1608218}
\end{align}
which models functions vanishing at $n=1/2$.
In other words, \eqref{1608218} may be understood as 
arising from sampling a continuous function $f$ at the points $n+1/2$:
$x[n]=f(n+1/2)$.
In such a model the Neumann condition is given by 
\begin{align*}
x[1]-x[0]=0,
\end{align*}
and the Robin condition by 
\begin{align*}
a(x[0]+x[1])/2+b(x[1]-x[0])=0;\quad (a,b)\neq (0,0).
\end{align*}

In any case all the above boundary conditions are unified as
\begin{align*}
x[0]=\alpha x[1];\quad \alpha\in\mathbb R.
\end{align*}
Denote the corresponding Laplacian by $H_\alpha$, i.e., 
for any sequence $x\colon\mathbb N\to\mathbb C$ 
\begin{equation}\label{N-formula}
(H_\alpha x)[n]=\begin{cases}
(2-\alpha)x[1]-x[2]& \text{for }n=1,\\
2x[n]-x[n+1]-x[n-1] & \text{for }n\geq2.
\end{cases}
\end{equation}
We note that the operator $H_\alpha$ is in fact 
bounded and self-adjoint on $\mathcal H=\ell^2(\mathbb N)$.

Let $e_1=(1,0,0,\ldots)$ be the first canonical basis vector 
and define the potential
\begin{equation}
V_\alpha=-\alpha\ket{e_1}\bra{e_1}.
\end{equation}
Then, comparing definitions \eqref{D-formula} and \eqref{N-formula}, we see that
\begin{equation}
H_\alpha=H_0+V_\alpha.
\end{equation}
The potential $V_\alpha$ satisfies Assumption~\ref{assumV} with $\mathcal{K}=\mathbb{C}$ and
\begin{align}
v=\sqrt{|\alpha|}\ket{e_1},\quad
v^{\ast}=\sqrt{|\alpha|}\bra{e_1},\quad
U=-\mathop{\mathrm{sgn}}\alpha.
\label{16082115}
\end{align}
Actually $V_\alpha$ is a multiplication operator.
We can directly compute 
\begin{align*}
\widetilde{\mathcal E}
=\mathbb C((1-\alpha)\mathbf n+\alpha\mathbf 1),\quad
\mathsf E=\{0\}.
\end{align*}
Note that these eigenspaces can also be computed by 
applying the results of Section~\ref{12.12.19.2.5} to \eqref{16082115}.
The above description of the eigenspaces implies the following:

\begin{lemma}\label{1608211512}
The threshold $0$ for the operator $H_\alpha$ is 
\begin{enumerate}
\item
a regular point if $\alpha\neq 1$;
\item
an exceptional point of the first kind if $\alpha=1$.
\end{enumerate}
\end{lemma}

We can construct the Fourier transform associated with $H_\alpha$,
and compute its expansion coefficients explicitly, which of course 
coincide with those computed from 
Theorems~\ref{thm-reg}--\ref{thm-ex3} and Lemma~\ref{1608211512}.
We remark that we may choose the Neumann Laplacian as the free operator, 
instead of the Dirichlet Laplacian, and formulate our main results 
for its perturbations. 
However, then the proofs get much 
more complicated,  
since its threshold $0$ is an exceptional point of the first kind,
which otherwise is regular.

\section{Generalized eigenspaces}\label{12.12.19.2.5}

In this section we write down the eigenspaces using subspaces of $\mathcal K$,
and then derive some useful properties.
In particular, we reveal the relation between invertibility of 
\textit{intermediate operators} and threshold types.
Compared with the full line discussed in Ito-Jensen\cite{IJ},
the half-line has a very clear correspondence between them,
and the threshold structure is much simpler. 
This is because the free resolvent on the half-line does not have a singular term,
and hence that of the perturbed resolvent comes only 
and directly from those intermediate operators.

To state the main results of this section let us introduce some notation. 
Let 
\begin{align}
M_0=U+v^*G_{0,0}v,\quad M_1=v^*G_{0,1}v
=-|v^*\mathbf n\rangle\langle v^*\mathbf n|,
\label{160817218}
\end{align}
and $Q,S\in\mathcal B(\mathcal K)$ be 
the orthogonal projections onto 
$\mathop{\mathrm{Ker}}M_0,\mathop{\mathrm{Ker}}M_1$, respectively.
Then we set 
\begin{align}
m_0=QM_1Q=-\ket{Qv^*\mathbf{n}}\bra{Qv^*\mathbf{n}}.
\label{160817219}
\end{align}
The operators $M_0$ and $m_0$ are, so to say, the 
\textit{intermediate operators} 
in the terminology of Ito-Jensen\cite{IJ} for the half-line case. 
They actually appear as expansion coefficients of certain operators in the 
later sections, but at least here we can define them independently of 
these expansions. They are well-defined 
for any $\beta\ge 1$ in Assumption~\ref{assumV}.
In addition, we also define the operators $w\in\mathcal B((\mathcal L^\beta)^*,\mathcal K)$
and $z\in \mathcal B(\mathcal K, \mathcal L^*)$ by 
\begin{align}
w=Uv^*,\quad 
z
=
\|v^*\mathbf n\|^{\dagger 2}
\langle M_0v^*\mathbf n,{}\cdot{}\rangle \mathbf n 
-G_{0,0}v,
\label{12.12.27.16.4}
\end{align}
where $a^{\dagger}$ denotes 
the pseudo-inverse of $a\in\mathbb C$, see \eqref{13.3.5.22.28}.

\begin{proposition}\label{13.1.16.2.51}
Suppose that $\beta\ge 1$ in Assumption~\ref{assumV}.
Then the eigenspaces are expressed as 
\begin{align}
\widetilde{\mathcal E}&=
z(\mathop{\mathrm{Ker}} S M_0)\oplus
\bigl(\mathbb C\mathbf n \cap \mathop{\mathrm{Ker}} v^*\bigr),
\label{12.12.19.6.56}\\
\mathcal E
&=
z(\mathop{\mathrm{Ker}}M_0)
,
\label{12.12.19.6.57}\\
\mathsf E&
=z(\mathop{\mathrm{Ker}}M_0\cap\mathop{\mathrm{Ker}} M_1)
=z(\mathop{\mathrm{Ker}}M_0\cap\mathop{\mathrm{Ker}} m_0).
\label{12.12.19.6.58}
\end{align}
In particular, the generalized eigenfunctions 
have the special asymptotics \eqref{13.3.7.13.48},
and, also, a function $\Psi\in\mathcal E$ has the
asymptotics $\Psi-\mathbf 1\in \mathcal L^{\beta-2}$
if and only if $\langle V\mathbf n,\Psi\rangle=-1$.
\end{proposition}

\begin{corollary}\label{160616438}
Suppose that $\beta\ge 1$ in Assumption~\ref{assumV}.
\begin{enumerate}
\item\label{12.12.19.6.30}
The threshold $0$ is a regular point if and only if 
$M_0$ is invertible in $\mathcal B(\mathcal K)$.
In addition, if the threshold $0$ is a regular point, 
\begin{align*}
\dim (\widetilde{\mathcal E}/\mathcal E)=1,\quad 
\dim(\mathcal E/\mathsf E)=
\dim\mathsf E=0
.
\end{align*}

\item\label{12.12.19.6.31}
The threshold $0$ is an exceptional point of the first kind 
if and only if 
$M_0$ is not invertible in $\mathcal B(\mathcal K)$
and $m_0$ is invertible in $\mathcal B(Q\mathcal K)$.
In addition, if the threshold $0$ is an exceptional point of the first kind, 
\begin{align*}
\dim (\widetilde{\mathcal E}/\mathcal E)=0,\quad 
\dim(\mathcal E/\mathsf E)=1,\quad
\dim\mathsf E=0
.
\end{align*}

\item\label{12.12.19.6.32}
The threshold $0$ is an exceptional point of the second kind 
if and only if 
$M_0$ is not invertible in $\mathcal B(\mathcal K)$
and $m_0=0$.
In addition, if the threshold $0$ is an exceptional point of the second kind, 
\begin{align*}
\dim (\widetilde{\mathcal E}/\mathcal E)=1,\quad 
\dim(\mathcal E/\mathsf E)=0,\quad
1\le\dim\mathsf E<\infty
.
\end{align*}

\item\label{12.12.19.6.33}
The threshold $0$ is an exceptional point of the third kind 
if and only if 
$M_0$ and $m_0$ are not invertible in $\mathcal B(\mathcal K)$ and 
$\mathcal B(Q\mathcal K)$, respectively, and $m_0\neq 0$.
In addition, if the threshold $0$ is an exceptional point of the third kind, 
\begin{align*}
\dim (\widetilde{\mathcal E}/\mathcal E)=0,\quad 
\dim(\mathcal E/\mathsf E)=1,\quad
1\le\dim\mathsf E<\infty
.
\end{align*}
\end{enumerate}
\end{corollary}

\begin{corollary}\label{13.3.25.13.0}
Suppose that $\beta\ge 1$ in Assumption~\ref{assumV},
and that $V$ is local.
Then
\begin{align}
\dim \widetilde{\mathcal E}=1,\quad 
\dim\mathsf E=0,
\label{16081618}
\end{align}
i.e., the threshold $0$ is either a regular point or 
an exceptional point of the first kind.
\end{corollary}

In the remainder of this section
we prove Proposition~\ref{13.1.16.2.51},
and Corollaries~\ref{160616438} and \ref{13.3.25.13.0},
using a sequence of lemmas given below.

\begin{lemma}\label{12.11.24.18.24}
For any $x\in\mathcal L^s$, $s\ge 1$,
the sequence $G_{0,0}x\in \mathcal L^*$ is expressed as 
\begin{align}
(G_{0,0}x)[n]
&=
\langle\mathbf n,x\rangle
-\sum_{m=n}^\infty(m-n)x[m]\quad \text{for } n\in\mathbb N.
\label{12.11.24.16.17}
\end{align}
In particular, $G_{0,0}x\in \mathcal L^{s-2}$ 
if and only if $\langle \mathbf n,x\rangle=0$.
\end{lemma}
\begin{proof}
By \eqref{G00} we can write  
\begin{align*}
(G_{0,0}x)[n]
&=\sum_{m=1}^{n-1}mx[m]
+\sum_{m=n}^\infty nx[m],
\end{align*}
which immediately implies \eqref{12.11.24.16.17}.
Noting that 
\begin{align*}
\sum_{n=1}^\infty (1+n^2)^{(s-2)/2}
\biggl|\sum_{m=n}^\infty (m-n)x[m]\biggr|
\le C\|x\|_{1,s}<\infty,
\end{align*}
we can deduce that the second term 
on the right-hand side of \eqref{12.11.24.16.17} 
belongs to $\mathcal L^{s-2}$.
Then by the fact that $\mathbf 1\notin \mathcal L^{s-2}$ for $s\ge 1$ 
we can verify the last assertion.
\end{proof}

\begin{lemma}\label{13.1.18.6.0}
The compositions $H_0G_{0,0}$ and $G_{0,0}H_0$,
defined on $\mathcal L^1$ and $\mathbb C\mathbf n\oplus 
\mathbb C\mathbf 1\oplus
\mathcal L^1$, respectively, 
are expressed as 
\begin{align*}
H_0G_{0,0}=I_{\mathcal L^1},\quad
G_{0,0}H_0=\Pi,%\label{13.1.18.20.17}
\end{align*}
where $\Pi\colon 
\mathbb C\mathbf n\oplus 
\mathbb C\mathbf 1\oplus
\mathcal L^1
\to 
\mathbb C\mathbf 1\oplus
\mathcal L^1$ is the projection.
\end{lemma}
\begin{remark}
Lemmas~\ref{12.11.24.18.24} and \ref{13.1.18.6.0} in particular imply that 
for any $s\ge 1$
\begin{align}
\begin{split}
\mathbb C\mathbf 1\oplus
\mathcal L^s
\subset 
G_{0,0}(\mathcal L^s)
\subset 
\mathbb C\mathbf 1\oplus
\mathcal L^{s-2}.
\end{split}\label{13.1.16.1.52}
\end{align}
\end{remark}
\begin{proof}
By direct computation employing the expression \eqref{12.11.24.16.17} 
we can verify that 
for any $x\in\mathcal L^1$
$$H_0G_{0,0}x=G_{0,0}H_0x=x.$$ 
We can also compute 
\begin{align*}
H_0\mathbf n=0,\quad
G_{0,0}H_0\mathbf 1=\mathbf 1.
\end{align*}
Then the assertion follows by the above identities.
\end{proof}

\begin{lemma}\label{12.12.27.16.6}
For any $\Phi\in \mathop{\mathrm{Ker}}S M_0$
and $\Psi\in\widetilde{\mathcal E}$
\begin{align}
wz\Phi
=\Phi,\quad
zw\Psi\in\widetilde{\mathcal E}
.\label{12.11.26.19.10}
\end{align}
In addition, 
\begin{align}
z^{-1}(\widetilde{\mathcal E})&=\mathop{\mathrm{Ker}} S M_0,
&
\widetilde{\mathcal E}\cap 
\mathop{\mathrm{Ker}}w
&=\mathbb C\mathbf n \cap \mathop{\mathrm{Ker}}v^*,
\label{12.11.26.19.0}\\
z^{-1}(\mathcal E)
&=\mathop{\mathrm{Ker}}M_0,&
\mathcal E\cap 
\mathop{\mathrm{Ker}}w
&=\{0\},
\label{12.11.26.19.1}\\
z^{-1}(\mathsf E)
&=\mathop{\mathrm{Ker}}M_0\cap \mathop{\mathrm{Ker}} M_1,&
\mathsf E\cap \mathop{\mathrm{Ker}}w
&=\{0\}.
\label{12.11.27.2.28}
\end{align}
\end{lemma}
\begin{proof}
\textit{Step 1.}\quad
We prove the first assertion of \eqref{12.11.26.19.10}.
Let $\Phi\in \mathop{\mathrm{Ker}}S M_0$.
Then, using $v^{\ast}G_{0,0}v=M_0-U$, we can compute 
\begin{align*}
w z\Phi
&=
Uv^*\Bigl[
\|v^*\mathbf n\|^{2\dagger}\langle M_0v^*\mathbf n,\Phi\rangle \mathbf n
-G_{0,0}v\Phi\Bigr]\\
&= U(1-S)M_0\Phi-UM_0\Phi+\Phi\\
&=\Phi.
\end{align*}

\smallskip
\noindent
\textit{Step 2.}\quad
Before the second assertion of \eqref{12.11.26.19.10}
we prove \eqref{12.11.26.19.0}.
We first note that by Lemma~\ref{13.1.18.6.0} and $v^{\ast}G_{0,0}v=M_0-U$
for any $\Phi \in \mathcal K$ 
\begin{align}
Hz\Phi
&=
(H_0+vUv^*)\Bigl[\|v^*\mathbf n\|^{2\dagger}
\langle M_0v^*\mathbf n,\Phi\rangle \mathbf n -G_{0,0}v\Phi\Bigr]\notag
\\&
=
-v\Phi\notag\\
&\quad+\|v^*\mathbf n\|^{2\dagger}\langle M_0v^*\mathbf n,\Phi\rangle vUv^*\mathbf n -vU(M_0-U)\Phi\notag\\
&=-vUS M_0\Phi.
\label{16081614}
\end{align}
Then, since $vU$ is injective,
it follows that $z\Phi\in \widetilde{\mathcal E}$ if and only if
$\Phi\in\mathop{\mathrm{Ker}}S M_0$,
which implies the first identity of (\ref{12.11.26.19.0}).
As for the second, 
we first note that 
for any $\Psi\in \widetilde{\mathcal E}\cap \mathop{\mathrm{Ker}}w$
\begin{align*}
H_0\Psi =0,\quad
v^*\Psi=0.
\end{align*}
Since the first identity $H_0\Psi =0$ 
can be rephrased as $\Psi\in \mathbb C\mathbf n $,
we obtain $\Psi \in \mathbb C\mathbf n \cap \mathop{\mathrm{Ker}}v^*$.
The inverse inclusion is almost obvious,
and hence the second identity of (\ref{12.11.26.19.0}).

\smallskip
\noindent
\textit{Step 3.}\quad
Now we prove the second assertion of \eqref{12.11.26.19.10}.
Let $\Psi\in\widetilde{\mathcal E}$. 
Then by reusing \eqref{16081614} and noting 
$M_0=U+v^{\ast}G_{0,0}v$ and Lemma~\ref{13.1.18.6.0}
\begin{align*}
Hzw\Psi
&
=-vUS (v^*+v^{\ast}G_{0,0}V)\Psi
\\&
=-vUS v^*G_{0,0}(H_0+V)\Psi
\\&
=0,
\end{align*}
which implies $zw\Psi\in\widetilde{\mathcal E}$.

\smallskip
\noindent
\textit{Step 4.}\quad
Let us prove (\ref{12.11.26.19.1}).
Let $\Phi\in \mathcal K$.
By Lemma~\ref{12.11.24.18.24} we can write
\begin{align}
\begin{split}
z\Phi[n]
&=
\|v^*\mathbf n\|^{2\dagger}\langle v^*\mathbf n,M_0\Phi\rangle\mathbf n [n]
-\langle v^*\mathbf n ,\Phi\rangle\mathbf 1[n]
\\&\quad
+\sum_{m=n}^\infty (m-n)(v\Phi)[m].
\end{split}\label{b11.1.30.19.38}
\end{align}
As in the proof of Lemma~\ref{12.11.24.18.24},
the last term in \eqref{b11.1.30.19.38}
belong to $\mathcal L^{\beta-2}$.
This fact combined with the first identity of (\ref{12.11.26.19.0})
implies that $z\Phi\in \mathcal E$
if and only if 
\begin{align*}
\Phi\in\mathop{\mathrm{Ker}}S M_0,\quad
\|v^*\mathbf n\|^{2\dagger}\langle v^*\mathbf n,M_0\Phi\rangle =0.
\end{align*}
Hence the first identity of (\ref{12.11.26.19.1}) is obtained.
As for the second one we can proceed as in Step 2, and 
it is almost obvious.

\smallskip
\noindent
\textit{Step 5.}\quad
The assertion (\ref{12.11.27.2.28}) 
can be shown similarly to Step 4,
and we omit the details.
\end{proof}

\begin{proof}[Proof of Proposition~\ref{13.1.16.2.51}.]
From (\ref{12.11.26.19.10})
and the first identity of (\ref{12.11.26.19.0}) 
we can deduce that the restrictions
\begin{align*}
z|_{\mathop{\mathrm{Ker}}S M_0}
\colon\mathop{\mathrm{Ker}}S M_0\to \widetilde{\mathcal E},\quad
w|_{\widetilde{\mathcal E}}\colon \widetilde{\mathcal E}\to \mathop{\mathrm{Ker}}S M_0
\end{align*}
are injective and surjective, respectively. 
Hence, the asserted isomorphisms (\ref{12.12.19.6.56})--(\ref{12.12.19.6.58}) 
are direct consequences of (\ref{12.11.26.19.0})--(\ref{12.11.27.2.28}),
respectively.
We note that the last inequality of \eqref{12.12.19.6.58}
is obvious by the definitions \eqref{160817218} and 
\eqref{160817219}.

The asymptotics \eqref{13.3.7.13.48} follows immediately by \eqref{12.12.19.6.56}, 
\eqref{12.12.27.16.4} and \eqref{13.1.16.1.52}.
Next, for any $\Psi\in\mathcal E$ we let $\Phi=w\Psi=Uv^*\Psi\in 
\mathop{\mathrm{Ker}}M_0$.
Then, since $\Psi=z\Phi=-G_{0,0}v\Phi$,  
Lemma~\ref{12.11.24.18.24} implies that 
$\Psi-\mathbf 1\in\mathcal L^{\beta-2}$ if and only if 
$\langle \mathbf n,-v\Phi\rangle=1$,
which in turn is equivalent to $\langle V\mathbf n,\Psi\rangle=-1$.
Hence we are done.
\end{proof}

\begin{proof}[Proof of Corollary~\ref{160616438}]
We first claim that 
\begin{align}
\dim (\widetilde{\mathcal E}/\mathcal E)\le 1,\quad 
\dim(\mathcal E/\mathsf E)\le 1,\quad 
\dim\mathsf E<\infty.
\label{13.3.7.13.49}
\end{align}
The first and second inequalities of \eqref{13.3.7.13.49} are obvious by 
\eqref{13.3.7.13.48}, \eqref{16071616} and \eqref{16071617}.
For the last inequality of \eqref{13.3.7.13.49}
we note that $Uv^*G_{0,0}v\in \mathcal C(\mathcal K)$.
Then
\begin{align*}
\dim \mathsf E
&\le \dim \mathcal E
=\dim\mathop{\mathrm{Ker}} M_0\\
&=\dim\mathop{\mathrm{Ker}} (1+Uv^*G_{0,0}v)<\infty.
\end{align*}
Hence the claim follows.

Now we prove the assertions \ref{12.12.19.6.30}--\ref{12.12.19.6.33}
of the corollary.
We note that the former parts of \ref{12.12.19.6.30}--\ref{12.12.19.6.33} 
are obvious by Proposition~\ref{13.1.16.2.51},
and hence we may discuss only the latter parts.

\smallskip
\noindent
\textit{\ref{12.12.19.6.30}.}\quad
Let the threshold $0$ be a regular point.
Then by definition we have
$$\dim\mathcal E=\dim\mathsf E=0.$$ 
If $v^*\mathbf n=0$, then,
since $S=I_{\mathcal K}$, we have by \eqref{12.12.19.6.56}
that 
$\widetilde{\mathcal E}
=\mathbb C\mathbf n$.
Otherwise, noting that $M_0$ is invertible, we have by \eqref{12.12.19.6.56}
that $\widetilde{\mathcal E}
=\mathbb CzM^{-1}v^*\mathbf n$.
In either cases we can conclude that 
$$\dim \widetilde{\mathcal E}=1.$$

\smallskip
\noindent
\textit{\ref{12.12.19.6.31}.}\quad
Let the threshold $0$ be an exceptional point of the first kind.
Then by definition and claim \eqref{13.3.7.13.49}
\begin{align*}
\dim\mathcal E=1,\quad
\dim\mathsf E=0.
\end{align*}
Let us show that $\widetilde{\mathcal E}=\mathcal E$.
Since 
$Q\mathcal K$ is nontrivial and 
$m_0=-\ket{Qv^*\mathbf{n}}\bra{Qv^*\mathbf{n}}$ is invertible there,
it follows that 
\begin{align}
Qv^*\mathbf n\neq 0. 
\label{1608172134}
\end{align}
Now it suffices to show that 
$\mathop{\mathrm{Ker}} S M_0\subset  \mathop{\mathrm{Ker}} M_0$.
Let $\Phi\in \mathop{\mathrm{Ker}} S M_0$.
Since $S$ is the orthogonal projections onto 
the kernel of $M_1$ given by \eqref{160817218},
there exists $c\in\mathbb C$ such that 
\begin{align*}
M_0\Phi=cv^*\mathbf n.
\end{align*}
Apply $Q$ to both sides above, 
then by \eqref{1608172134} it follows that $c=0$.
Hence $\Phi\in \mathop{\mathrm{Ker}}M_0$,
and the latter assertion is verified.

\smallskip
\noindent
\textit{\ref{12.12.19.6.32}.}\quad
Let the threshold $0$ be an exceptional point of the second kind.
Then by definition and claim \eqref{13.3.7.13.49}
\begin{align*}
\dim(\mathcal E/\mathsf E)=0,\quad
1\le \dim\mathsf E<\infty.
\end{align*}
If $v^*\mathbf n=0$, then $S=I_{\mathcal K}$,
and hence by \eqref{12.12.19.6.56}
\begin{align*}
\widetilde{\mathcal E}&=
z(\mathop{\mathrm{Ker}} M_0)\oplus
\mathbb C\mathbf n
=\mathcal E\oplus
\mathbb C\mathbf n.
\end{align*}
Otherwise, since $m_0=-\ket{Qv^*\mathbf{n}}\bra{Qv^*\mathbf{n}}=0$, 
we have 
$$0\neq v^*\mathbf n\in (\mathop{\mathrm{Ker}}M_0)^\perp
=\mathop{\mathrm{Ran}}M_0,$$
and hence we can find $\Phi\in \mathcal K\setminus \{0\}$ such that 
$M_0\Phi=v^*\mathbf n$.
Such $\Phi$ is unique up to $\mathop{\mathrm{Ker}}M_0$, and then
by \eqref{12.12.19.6.56}
\begin{align*}
\widetilde{\mathcal E}&=
z(\mathop{\mathrm{Ker}} M_0\oplus
\mathbb C\Phi)
=\mathcal E\oplus
\mathbb Cz\Phi.
\end{align*}
In either cases we obtain 
\begin{align*}
\dim(\widetilde{\mathcal E}/\mathcal E)=1.
\end{align*}

\smallskip
\noindent
\textit{\ref{12.12.19.6.33}.}\quad
Let the threshold $0$ be an exceptional point of the third kind.
Then by definition and claim \eqref{13.3.7.13.49}
\begin{align*}
\dim(\mathcal E/\mathsf E)=1,\quad
1\le \dim\mathsf E<\infty.
\end{align*}
Now it suffices to show that $\widetilde{\mathcal E}=\mathcal E$,
but this can be proved exactly the same manner as in
the proof of the assertion \ref{12.12.19.6.31} above.
Hence we are done.
\end{proof}

\begin{proof}[Proof of Corollary~\ref{13.3.25.13.0}]
It suffces to show that $\mathsf E=\{0\}$.
Let $\Psi\in \mathsf E$. 
Then it follows by Lemma~\ref{12.12.27.16.6} that $\Psi=zw\Psi$.
This equation can be rephrased as 
\begin{align}
\Psi[n]=\sum_{m=n}^\infty (m-n)V[m]\Psi[m]
\label{1608161808}
\end{align}
by Lemma~\ref{12.11.24.18.24} and 
the asymptotics of $\Psi$ as $n\to \infty$.
Since $V\in \mathcal L^\beta$, we can 
choose large $n_0\ge 0$ such that 
\begin{align}
\sum_{n=n_0}^\infty n|V[n]|\le \tfrac12.
\label{1608161809}
\end{align}
By \eqref{1608161808} and \eqref{1608161809} we obtain
\begin{align*}
\bigl|\Psi[n]\bigr|&\le \tfrac12 \sup_{m\ge n_0}\bigl|\Psi[m]\bigr|
\text{ for }n\ge n_0,\\
\intertext{or}  
\Psi[n]&=0\text{ for }n\ge n_0.
\end{align*}
Since the equation $H\Psi=0$ is a difference equation,
the above initial condition at infinity
yields $\Psi=0$, and hence $\mathsf E=\{0\}$.
Hence we are done.
\end{proof}

\section{The first step in resolvent expansion}\label{1608217}

This section gives a short preliminary computation for the proofs 
of Theorems~\ref{thm-reg}--\ref{thm-ex3} given in the 
following sections.
These computations are common to all the proofs.

Define the operator
$M(\kappa)\in\mathcal B(\mathcal K)$
for $\mathop{\mathrm{Re}}\kappa>0$ by
\begin{equation}
M(\kappa)=U+v^*R_0(\kappa)v.\label{Mdef}
\end{equation}
Fix $\kappa_0>0$ such that $z=-\kappa^2$ belongs to the resolvent set of $H$ for any $\mathop{\mathrm{Re}}\kappa\in(0,\kappa_0)$. This is possible due to the decay assumptions on $V$.

\begin{lemma}\label{prop23}
Let the operator $M(\kappa)$ be defined as above.
\begin{enumerate}
\item\label{16082118}
Let Assumption~\ref{assumV} hold for some integer $\beta\ge 2$.
Then 
\begin{equation}\label{M-expand}
M(\kappa)=\sum_{j=0}^{\beta-2} \kappa^jM_j+\mathcal{O}(\kappa^{\beta-1})
\quad 
\text{in }\mathcal B(\mathcal K)
\end{equation}
with $M_j\in\mathcal B(\mathcal K)$ given by
\begin{align}\label{m-expand}
M_0=U+v^*G_{0,0}v,\quad
M_j=v^*G_{0,j}v\text{ for }  j\geq1.
\end{align}

\item\label{16082119}
Let Assumption~\ref{assumV} hold with $\beta\ge 1$.
For any $0<\mathop{\mathrm{Re}}\kappa<\kappa_0$ the operator $M(\kappa)$
is invertible in $\mathcal B(\mathcal K)$,
and 
\begin{align*}
M(\kappa)^{-1}&=U-Uv^*R(\kappa)vU.
\end{align*}
Moreover, 
\begin{align}
R(\kappa)&=R_0(\kappa)
-R_0(\kappa)vM(\kappa)^{-1}v^*R_0(\kappa).\label{second-resolvent}
\end{align}
\end{enumerate}
\end{lemma}
\begin{proof}
\textit{1.}\quad 
This result follows from Assumption~\ref{assumV} and Proposition~\ref{prop12}.

\smallskip
\noindent
\textit{2.}\quad 
The assertion is verified by direct computations, 
see Ito-Jensen\cite[Proposition 1.13]{IJ}.
\end{proof}

Note that the operators $M_0$ and $M_1$ coincide with those defined 
in Section~\ref{12.12.19.2.5}.

By 
Lemma~\ref{prop23}.\ref{16082118}
the operator $M(\kappa)$ has an expansion, 
and 
by Lemma~\ref{prop23}.\ref{16082119} and Proposition~\ref{prop12} 
an expansion of $R(\kappa)$ is reduced to that of the inverse 
$M(\kappa)^{-1}$. 
If the leading operator $M_0\in\mathcal B(\mathcal K)$ 
is invertible, or by Proposition~\ref{13.1.16.2.51}, 
if the threshold $0$ is a regular point, 
we can employ the Neumann series to compute the expansion of $M(\kappa)^{-1}$.
Otherwise, we shall employ an inversion formula introduced in Jensen-Nenciu\cite{JN0} in 
a way similar to Ito-Jensen\cite{IJ}. 
We note that we are also going to use the pseudo-inverse several times.
For reference 
we present the inversion formula and the pseudo-inverse 
in Appendix~\ref{1608211941}.

\section{Regular threshold}\label{1608241726}

In this section we prove Theorem~\ref{thm-reg}.
In this case the leading operator $M_0$ in the expansion \eqref{M-expand}
is invertible by Corollary~\ref{160616438}.
Hence the inversion formula in Appendix~\ref{1608211941} is not needed.

\begin{proof}[Proof of Theorem~\ref{thm-reg}]
By the assumption and Corollary~\ref{160616438}
it follows that $M_0$ is invertible in $\mathcal{B}(\mathcal{K})$.
Hence we can use the Neumann series to invert \eqref{M-expand}.
Let us write it as
\begin{align}
M(\kappa)^{-1}=\sum_{j=0}^{\beta-2}\kappa^jA_j+\mathcal O(\kappa^{\beta-1}),\quad
A_j\in\mathcal B(\mathcal K).
\label{160819}
\end{align}
The coefficients $A_j$ are written explicitly in terms of the $M_j$.
The first two terms are
\begin{align}
A_0&=M_0^{-1},\quad 
A_1=-M_0^{-1}M_1M_0^{-1}.\label{12.12.1.7.18}
\end{align}
We insert the expansions \eqref{free-expan} with $N=\beta-2$ and \eqref{160819}
into \eqref{second-resolvent}, and then obtain the expansion
\begin{align*}
R(\kappa)&=\sum_{j=0}^{\beta-2}\kappa^j
G_j+\mathcal O(\kappa^{\beta-1});\\
G_j&=G_{0,j}
-\sum_{\genfrac{}{}{0pt}{}{j_1\ge 0,j_2\ge 0,j_3\ge 0}{j_1+j_2+j_3=j}}
G_{0,j_1}vA_{j_2}v^*G_{0,j_3}.
\end{align*}
This result and \eqref{12.12.1.7.18} in particular leads to the expressions 
\begin{align*}
G_0&=G_{0,0}-G_{0,0} vM_0^{-1}v^*G_{0,0}
,
\\
G_1
&=G_{0,1}
-G_{0,1}vM_0^{-1}v^*G_{0,0}\\
&\quad+G_{0,0}vM_0^{-1}M_1M_0^{-1}v^*G_{0,0}
-G_{0,0}vM_0^{-1}v^*G_{0,1}
\\&
=
(I-G_{0,0} vM_0^{-1}v^*)G_{0,1}(I-vM_0^{-1}v^*G_{0,0}).
\end{align*}
The expression \eqref{160822427} is obtained.
The expression \eqref{160822428} follows by noting
$$(I-G_{0,0} vM_0^{-1}v^*)\mathbf n=\widetilde\Psi_c,$$
which can be verified with ease by \eqref{12.12.19.6.56}.
\end{proof}
\begin{proof}[Verification of \eqref{16082620}]
The first identity in \eqref{16082620} follows by 
\begin{align*}
(I+&G_{0,0} V)(I-G_{0,0} vM_0^{-1}v^*)
\\
&=
I-G_{0,0} vM_0^{-1}v^* +G_{0,0} V-G_{0,0} VG_{0,0} vM_0^{-1}v^*
\\&
=I-G_{0,0} vU(U+v^*G_{0,0}v)M_0^{-1}v^* +G_{0,0} V
\\
&=I,
\end{align*}
\begin{align*}
(I-&G_{0,0} vM_0^{-1}v^*)(I+G_{0,0} V)\\
&=
I-G_{0,0} vM_0^{-1}v^*+G_{0,0} V-G_{0,0} vM_0^{-1}v^*G_{0,0} V
\\&
=I-G_{0,0} vM_0^{-1}(U+v^*G_{0,0}v)Uv^* +G_{0,0} V
\\
&=I.
\end{align*}
The second identity is verified analogously.
\end{proof}

\section{Exceptional threshold of the first kind}\label{1608241727}

In this section we prove Theorem~\ref{thm-ex1}. 
In this case the leading operator $M_0\in \mathcal B(\mathcal K)$ in \eqref{M-expand}
is not invertible, and 
we need the inversion formula given in Appendix~\ref{1608211941}
to invert the expansion~\eqref{M-expand}.

\begin{proof}[Proof of Theorem~\ref{thm-ex1}]
By the assumption and Corollary~\ref{160616438}
the leading operator $M_0$ from \eqref{M-expand}
is not invertible in $\mathcal B(\mathcal K)$, 
and we are going to apply Proposition~\ref{12.11.9.3.28}.
Let us write the expansion \eqref{M-expand} as
\begin{equation}\label{ex1-M-expand}
M(\kappa)
=\sum_{j=0}^{\beta-2}\kappa^jM_j+\mathcal{O}(\kappa^{\beta-1})=M_0+\kappa \widetilde{M}_1(\kappa).
\end{equation}
Let $Q$ be the orthogonal projection onto $\Ker M_0$, cf.\ 
Section~\ref{12.12.19.2.5}, and define 
\begin{equation}\label{def-J0}
m(\kappa)=
\sum_{j=0}^{\infty}(-1)^j
\kappa^jQ\widetilde{M}_1(\kappa)\bigl[(M_0^\dagger+Q)\widetilde{M}_1(\kappa)\bigr]^{j}Q.
\end{equation}
Then by Proposition~\ref{12.11.9.3.28} we have  
\begin{align}
M(\kappa)^{-1}&=(M(\kappa)+Q)^{-1}\notag\\
&\quad+\frac{1}{\kappa}
(M(\kappa)+Q)^{-1}m(\kappa)^{\dagger}(M(\kappa)+Q)^{-1}.
\label{16082255}
\end{align}
Note that by using \eqref{ex1-M-expand} we 
can rewrite \eqref{def-J0} in the form
\begin{equation}
m(\kappa)=\sum_{j=0}^{\beta-3}\kappa^jm_j+\mathcal{O}(\kappa^{\beta-2});
\quad 
m_j\in\mathcal B(Q\mathcal K).
\label{16082317}
\end{equation}
We have the following expressions for the first four coefficients:
\begin{align}
m_0&=QM_1Q,
\label{1608231743}
\\
m_1&=QM_2Q-QM_1(M_0^\dagger +Q)M_1Q,
\label{1608231744}
\\
m_2&=QM_3Q-QM_1(M_0^\dagger +Q)M_2Q\notag\\
&\quad-QM_2(M_0^\dagger +Q)M_1Q\notag
\\
&\quad+QM_1(M_0^\dagger +Q)M_1(M_0^\dagger +Q)M_1Q,
\label{1608231745}
\\
m_3&=QM_4Q
-QM_1(M_0^\dagger +Q)M_3Q\notag\\
&\quad-QM_2(M_0^\dagger +Q)M_2Q
-QM_3(M_0^\dagger +Q)M_1Q\notag
\\
&\quad
+QM_1(M_0^\dagger +Q)M_1(M_0^\dagger +Q)M_2Q\notag
\\
&\quad
+QM_1(M_0^\dagger +Q)M_2(M_0^\dagger +Q)M_1Q\notag
\\
&\quad
+QM_2(M_0^\dagger +Q)M_1(M_0^\dagger +Q)M_1Q\notag
\\
&\quad
-QM_1(M_0^\dagger +Q)M_1(M_0^\dagger +Q)M_1(M_0^\dagger +Q)M_1Q.
\label{1608231746}
\end{align}
Then by the assumption and Corollary~\ref{160616438} 
the coefficient $m_0=QM_1Q$ is invertible in $\mathcal B(Q\mathcal K)$.
Thus the Neumann series provides 
the expansion of the inverse $m(\kappa)^\dagger$.
Let us write it as 
\begin{align}
m(\kappa)^\dagger&=\sum_{j=0}^{\beta-3}\kappa^jA_j+\mathcal{O}(\kappa^{\beta-2}),\label{ex1-m-expand}\\
A_0&=m_0^\dagger,\quad A_j\in\mathcal B(Q\mathcal K).
\notag 
\end{align}
The Neumann series also provide an expansion of 
$(M(\kappa)+Q)^{-1}$, which we write as 
\begin{align}
(M(\kappa)+Q)^{-1}
=\sum_{j=0}^{\beta-2}\kappa^jB_j+\mathcal{O}(\kappa^{\beta-1}),
\label{1608225}
\end{align}
where $B_j\in\mathcal B(\mathcal K)$.
The first three coefficients can be written as follows: 
\begin{align*}
B_0&=M_0^\dagger +Q,\\
B_1&=-(M_0^\dagger +Q) M_1(M_0^\dagger +Q),\\
B_2&=-(M_0^\dagger +Q) M_2(M_0^\dagger +Q)\\
&\quad+(M_0^\dagger +Q) M_1(M_0^\dagger +Q) M_1(M_0^\dagger +Q).
\end{align*} 
Now we insert the expansions \eqref{ex1-m-expand} and \eqref{1608225} 
into the formula \eqref{16082255},
and then 
\begin{align}
M(\kappa)^{-1}
&=\sum_{j=-1}^{\beta-4}\kappa^jC_j+\mathcal O(\kappa^{\beta-3}),
\notag\\
C_j&=B_j
+\sum_{\genfrac{}{}{0pt}{}{j_1\ge 0,j_2\ge 0,j_3\ge 0}{j_1+j_2+j_3=j+1}}
B_{j_1}A_{j_2}B_{j_3},
\label{160822553}
\end{align}
with $B_{-1}=0$. 
Next we insert the expansions \eqref{free-expan} with $N=\beta-3$ and 
\eqref{160822553} into the formula \eqref{second-resolvent}.
Then we obtain the expansion
\begin{align*}
R(\kappa)&=\sum_{j=-1}^{\beta-4}\kappa^j
G_j+\mathcal O(\kappa^{\beta-3}),\\
G_j&=G_{0,j}
-\sum_{\genfrac{}{}{0pt}{}{j_1\ge 0,j_2\ge -1,j_3\ge 0}{j_1+j_2+j_3=j}}
G_{0,j_1}vC_{j_2}v^*G_{0,j_3},
\end{align*}
with $G_{0,-1}=0$. This verifies \eqref{expand-first}.

Next we compute $G_{-1}$.
By the above expressions we can write
\begin{align*}
G_{-1}&=-G_{0,0}vC_{-1} v^*G_{0,0}
=-G_{0,0}vm_0^\dagger v^*G_{0,0}
,
\end{align*}
and by \eqref{G01}
\begin{equation}\label{SM1S}
m_0=QM_1Q=-\ket{Qv^*\mathbf{n}}\bra{Qv^*\mathbf{n}}.
\end{equation}
The expression \eqref{SM1S} implies that 
$m_0$ is at most of rank $1$, but 
by the assumption and  Corollary~\ref{160616438} it is also invertible in 
$\mathcal B(Q\mathcal K)$.
Hence it follows that 
\begin{align*}
Qv^*\mathbf{n} \neq 0,\quad 
\dim\Ker M_0=\dim Q\mathcal K=1.
\end{align*}
Then we can write 
\begin{align}
m_0^\dagger &=- |\Phi_c\rangle\langle \Phi_c|;\\ 
\Phi_c&=-\|Qv^*\mathbf{n}\|^{-2}Qv^*\mathbf{n}\in Q\mathcal K=\Ker M_0,
\label{16082516}
\end{align}
such that 
\begin{align*}
G_{-1}=|\Psi_c\rangle\langle\Psi_c|; \quad
\Psi_c=-G_{0,0}v\Phi_c\in \mathcal E.
\end{align*}Let us to show that 
the above resonance function $\Psi_c$ is canonical.
We have 
\begin{align*}
\langle V\mathbf n,\Psi_c\rangle
&
=-\langle v^*\mathbf n,U(M_0-U)\Phi_c\rangle
=\langle v^*\mathbf n,\Phi_c\rangle
=-1,
\end{align*}
and hence we obtain \eqref{1608229}.

Finally we prove \eqref{16082513}. 
We first express $G_0$ by $A_*$ and $B_*$,
and then insert expressions for them:
\begin{align*}
G_0
&
=
G_{0,0}
-G_{0,0}vC_{-1}v^*G_{0,1}
-G_{0,1}vC_{-1}v^*G_{0,0}\\
&\quad
-G_{0,0}vC_0v^*G_{0,0}
\\&
=
G_{0,0}
-G_{0,0}vA_0v^*G_{0,1}
-G_{0,1}vA_0v^*G_{0,0}
\\&\quad
-G_{0,0}v\bigl(
B_0
+B_0A_0B_1
+B_1A_0B_0\\
&\qquad
+B_0A_1B_0
\bigr)v^*G_{0,0}
\\&
= 
G_{0,0}
-G_{0,0}vm_0^\dagger v^*G_{0,1}
-G_{0,1}vm_0^\dagger v^*G_{0,0}
\\&\quad
-G_{0,0}v\Bigl(
M_0^\dagger +Q
-m_0^\dagger M_1(M_0^\dagger +Q)\\
&\qquad
-(M_0^\dagger +Q) M_1m_0^\dagger
-m_0^\dagger m_1m_0^\dagger
\Bigr)v^*G_{0,0}.
\end{align*}
We expand the terms in big parentheses and unfold $m_1$,
noting $m_0m_0^\dagger=m_0^\dagger m_0=Q$: 
\begin{align*}
G_0&
= 
G_{0,0}
-G_{0,0}vm_0^\dagger v^*G_{0,1}
-G_{0,1}vm_0^\dagger v^*G_{0,0}
\\&\quad
-G_{0,0}v\Bigl(
M_0^\dagger
-m_0^\dagger M_1M_0^\dagger 
-M_0^\dagger M_1m_0^\dagger\\
&\qquad
-m_0^\dagger M_2m_0^\dagger
+m_0^\dagger M_1M_0^\dagger M_1m_0^\dagger
\Bigr)v^*G_{0,0}
\\&
= 
G_{0,0}
+G_{0,0}vm_0^\dagger M_2m_0^\dagger v^*G_{0,0}
-G_{0,0}vm_0^\dagger v^*G_{0,1}\\
&\quad-G_{0,1}vm_0^\dagger v^*G_{0,0}
\\&\quad
-G_{0,0}v\bigl(I-m_0^\dagger M_1\bigr)M_0^\dagger\bigl(I- M_1m_0^\dagger\bigr)v^*G_{0,0}.
\end{align*}
Now we use \eqref{16082516} and the expressions 
$M_j=v^*G_{0,j}v$, $j\ge 1$, and $G_{0,1}=-|\mathbf n\rangle\langle\mathbf n|$:
\begin{align*}
G_0
&
=
G_{0,0}
+|\Psi_c\rangle\langle v\Phi_c, G_{0,2}v\Phi_c\rangle\langle \Psi_c|
-|\Psi_c\rangle\langle\mathbf n|
-|\mathbf n\rangle\langle\Psi_c|
\\&\quad
-\bigl(G_{0,0}-|\Psi_c\rangle\langle\mathbf n|\bigr)
vM_0^\dagger v^*
\bigl(G_{0,0}-|\mathbf n\rangle\langle \Psi_c|\bigr).
\end{align*}
Hence it remains to compute the coefficient of the second term
in the last expression.
We have by $\Phi_c=Uv\Psi$
\begin{align*}
\langle v\Phi_c, G_{0,2}v\Phi_c\rangle
=\langle V\Psi_c, G_{0,2}V\Psi_c\rangle
=\langle H_0\Psi_c, G_{0,2}H_0\Psi_c\rangle.
\end{align*}
Here we remark that we cannot directly use $G_{0,2}H_0=-G_{0,0}$,
since \eqref{18082023} holds as an extension from rapidly decaying functions,
while $\Psi_c$ is not decaying.
However, it suffices to subtract the leading asymptotics as follows.
\begin{align*}
\langle v\Phi_c,& G_{0,2}v\Phi_c\rangle\\
&=
\bigl\langle H_0(\Psi_c-\mathbf 1), G_{0,2}H_0\Psi_c\bigr\rangle
+(G_{0,2}H_0\Psi_c)[1]
\\&
=
-\bigl\langle (\Psi_c-\mathbf 1), G_{0,0}H_0\Psi_c\bigr\rangle\\
&\quad+\bigl(G_{0,2}H_0(\Psi_c-\mathbf 1)\bigr)[1]
+(G_{0,2}H_0\mathbf 1)[1]
\\&
=
-\bigl\langle (\Psi_c-\mathbf 1), G_{0,0}H_0(\Psi_c-\mathbf 1)\bigr\rangle\\
&\quad
-\overline{(G_{0,0}(\Psi_c-\mathbf 1))[1]}
\\&\quad
-\bigl(G_{0,0}(\Psi_c-\mathbf 1)\bigr)[1]
+(G_{0,2}H_0\mathbf 1)[1]
\\&
=
-\|\Psi_c-\mathbf 1\|^2
-2\mathop{\mathrm{Re}}(G_{0,0}(\Psi_c-\mathbf 1))[1]\\
&\quad+(G_{0,2}H_0\mathbf 1)[1]
.
\end{align*}
The last two terms are computed by using the explicit
expressions \eqref{G00} and \eqref{G02}.
Then we obtain \eqref{16082513}.
\end{proof}

\section{Exceptional threshold of the second kind}\label{1608241728}

Here we prove Theorem~\ref{thm-ex2}.
For the first part of the proof 
we can almost repeat the argument of the previous section,
but the second part is rather non-trivial.
In fact, we need the following lemma.
\begin{lemma}\label{lemma24}
Let $x_\nu\in\mathcal{L}^4$, $\nu=1,2$. Assume that
\begin{equation}\label{zero-cond}
\ip{\mathbf{n}}{x_\nu}=0,\quad \nu=1,2.
\end{equation}
Then one has that  $G_{0,0}x_\nu\in\mathcal{L}^2$, $\nu=1,2$, and that
\begin{equation}
\ip{x_1}{G_{0,2}x_2}=-\ip{G_{0,0}x_1}{G_{0,0}x_2}.
\end{equation}
\end{lemma}
\begin{proof}
We extend the sequences $x_\nu\in\mathcal L^4$, $\nu=1,2$, antisymmetrically 
to the whole line $\mathbb Z$ by letting
\begin{align*}
\widetilde x_\nu[n]=\mathop{\mathrm{sgn}}[n]x_\nu[|n|],\quad n\in\mathbb Z.
\end{align*}
Noting that 
the kernels $G_{0,0}[n,m]$ and $G_{0,2}[n,m]$ have the expressions
\begin{align*}
G_{0,0}[n,m]&=-\tfrac12\bigl(|n-m|-(n+m)\bigr),
\\
G_{0,2}[n,m]&=\tfrac1{12}\bigl(|n-m|-|n-m|^3\\
&\quad-(n+m)+(n+m)^3\bigr),
\end{align*}
we also define operators $\widetilde G_{0,0}$
and $\widetilde G_{0,2}$ mapping antisymmetric functions on $\mathbb Z$ to themselves
by the integral kernels
\begin{align}
\widetilde G_{0,0}[n,m]&=-\tfrac12|n-m|,\notag\\
\widetilde G_{0,2}[n,m]&=\tfrac1{12}\bigl(|n-m|-|n-m|^3\bigr),
\label{16082316}
\end{align}
respectively. 
Then it is easy to check that for $\nu=1,2$, $j=0,2$ and $n\ge 1$
\begin{align}
(G_{0,j}x_\nu)[n]=(\widetilde G_{0,j}\widetilde x_\nu)[n]
=-(\widetilde G_{0,j}\widetilde x_\nu)[-n].
\label{1608231}
\end{align}
On the other hand, the kernels \eqref{16082316}
are the same as the convolution kernels in Ito-Jensen\cite[equation (2.5)]{IJ},
and hence under assumption~\eqref{zero-cond} Ito-Jensen\cite[Lemma 4.16]{IJ} applies.
It follows that 
$\widetilde G_{0,0}\widetilde x_\nu\in\ell^{1,2}(\mathbb Z)$ and that 
\begin{align}
\langle \widetilde x_1,\widetilde G_{0,2}\widetilde x_2\rangle
=-\langle \widetilde G_{0,0}\widetilde x_1,\widetilde G_{0,0}\widetilde x_2\rangle.
\label{160823}
\end{align}
Then by \eqref{1608231} and \eqref{160823} the assertion follows.
\end{proof}

\begin{proof}[Proof of Theorem~\ref{thm-ex2}]
By the assumption and Corollary~\ref{160616438}
the leading operator $M_0$ from \eqref{M-expand} is not invertible
in $\mathcal B(\mathcal K)$. 
Write the expansion \eqref{M-expand} 
in the same form as \eqref{ex1-M-expand},
let $Q$ be the orthogonal projection onto $\Ker M_0$, and 
define $m(\kappa)$ by the same formula as \eqref{def-J0}.
Then by Proposition~\ref{12.11.9.3.28} we have the 
same formula as \eqref{16082255}. 
Again, $m(\kappa)$ defined by \eqref{def-J0}
has the same expansion \eqref{16082317}
with the same expressions \eqref{1608231743}--\eqref{1608231746}
for its coefficients, 
but this time we actually have
\begin{align}\label{am0b}
m_0=0,
\quad
m_1=QM_2Q,
\quad 
m_2=0.
\end{align}
In fact, by the assumption  we have 
\begin{align}
m_0=QM_1Q&=-\ket{Qv^*\mathbf{n}}\bra{Qv^*\mathbf{n}}=0,\notag\\
&\text{or}\quad  Qv^*\mathbf n=0,
\label{16082220}
\end{align}
and hence \eqref{am0b} follows by \eqref{m-expand}, 
\eqref{G01}, \eqref{16082220} and \eqref{1608231743}--\eqref{1608231746}.
Now we note that then the operator $m_1$ has to be invertible in 
$\mathcal B(Q\mathcal K)$.
Otherwise, we can apply Proposition~\ref{12.11.9.3.28} once more,
but this leads to a singularity of order 
$\kappa^{-j}$, $j\ge 3$, in the expansion of $R(\kappa)$,
which contradicts the self-adjointness of $H$.
Hence  the Neumann series 
provides an expansion of $m(\kappa)^\dagger$ of the form
\begin{equation}
m(\kappa)^\dagger=\sum_{j=-1}^{\beta-5}\kappa^jA_j+\mathcal{O}(\kappa^{\beta-4}),\quad
A_j\in\mathcal B(Q\mathcal K),
\label{ex1-m-expandb}
\end{equation}
with, e.g.\ 
\begin{align*}
A_{-1}&=m_1^{\dagger},\quad
A_0=-m_1^\dagger m_2m_1^\dagger
,\\
A_1&=-m_1^\dagger m_3m_1^\dagger+m_1^\dagger m_2m_1^\dagger m_2m_1^\dagger 
.
\end{align*}
These are actually simplified by \eqref{am0b} as
\begin{align}
&A_{-1}=m_1^{\dagger},\quad
A_0=0,\quad
A_1=-m_1^\dagger m_3m_1^\dagger
.
\label{1608252242}
\end{align}
The Neumann series also provides  an expansion of $(M(\kappa)+Q)^{-1}$ 
in the same form as \eqref{1608225}
with the same coefficients given there.
Now we insert the expansions \eqref{ex1-m-expandb} and \eqref{1608225} 
into the formula \eqref{16082255},
and then 
\begin{align}
M(\kappa)^{-1}
&=\sum_{j=-2}^{\beta-6}\kappa^jC_j+\mathcal O(\kappa^{\beta-5});
\notag\\
C_j&=B_j
+\sum_{\genfrac{}{}{0pt}{}{j_1\ge 0,j_2\ge -1,j_3\ge 0}{j_1+j_2+j_3=j+1}}
B_{j_1}A_{j_2}B_{j_3}
\label{160822553b}
\end{align}
with $B_{-2}=B_{-1}=0$.
We then insert the expansions \eqref{free-expan} with $N=\beta-4$ and \eqref{160822553b}
into the formula \eqref{second-resolvent}.
Finally we obtain the expansion
\begin{align*}
R(\kappa)&=\sum_{j=-2}^{\beta-6}\kappa^j
G_j+\mathcal O(\kappa^{\beta-5});\\
G_j&=G_{0,j}
-\sum_{\genfrac{}{}{0pt}{}{j_1\ge 0,j_2\ge -2,j_3\ge 0}{j_1+j_2+j_3=j}}
G_{0,j_1}vC_{j_2}v^*G_{0,j_3}
\end{align*}
with $G_{0,-2}=G_{0,-1}=0$.

Next we compute the coefficients.
We can use the above expressions of the coefficients 
to write 
\begin{align}
G_{-2}&
=
-G_{0,0}vC_{-2}v^*G_{0,0}\notag\\
&=
-G_{0,0}v m_1^{\dagger}v^*G_{0,0}\notag\\
&=z (Qv^*G_{0,2}vQ)^{\dagger}z^*.
\label{160826}
\end{align}
By this expression we can see that 
\begin{align*}
\mathop{\mathrm{Ran}}G_{-2}=(\mathop{\mathrm{Ker}}G_{-2})^\perp
\subset \mathcal E=\mathsf E.
\end{align*}
In addition, by Proposition~\ref{13.1.16.2.51} 
for any $\Psi\in \mathsf E$ we can write  
$\Psi=z\Phi=-G_{0,0}v\Phi$ 
for some $\Phi\in Q\mathcal K$, 
so that by Lemma~\ref{lemma24} 
\begin{align*}
\langle \Psi,G_{-2}\Psi\rangle
&=
-\langle G_{0,0}v\Phi,G_{0,0}v (Qv^*G_{0,2}vQ)^{\dagger}z^*\Psi\rangle\\
&=
\|\Psi\|_{\mathcal H}^2
.
\end{align*}
Since $G_{-2}$ is obviously self-adjoint on $\mathsf E$,
this implies that $G_{-2}$ coincides with the orthogonal projection
$P_0$ onto $\mathsf E$, as asserted in \eqref{ex2-G-2}.

As for $G_{-1}$, 
we can first write 
\begin{align*}
G_{-1}
&=
-G_{0,0}vC_{-1}v^*G_{0,0}\\
&\quad-G_{0,0}vC_{-2}v^*G_{0,1}
-G_{0,1}vC_{-2}v^*G_{0,0}.
%%\label{160822203}
\end{align*}
If we make use of the vanishing in \eqref{am0b}, \eqref{16082220}
and \eqref{1608252242},
we can easily verify \eqref{ex2-G-1} from this expression. 
We omit the details.

Next, we compute $G_1$.
Let us write, implementing $B_0A_*=A_*B_0=A_*$,
\begin{align*}
G_0
&=
G_{0,0}
-G_{0,0}vC_0v^*G_{0,0}\\
&\quad-G_{0,0}vC_{-1}v^*G_{0,1}
-G_{0,1}vC_{-1}v^*G_{0,0}
\\&\quad
-G_{0,0}vC_{-2}v^*G_{0,2}
-G_{0,1}vC_{-2}v^*G_{0,1}\\
&\quad-G_{0,2}vC_{-2}v^*G_{0,0}
\\&
=
G_{0,0}
-G_{0,0}v\bigl(B_0
+A_1
+A_0B_1
+B_1A_0\\
&\qquad+B_1A_{-1}B_1
+A_{-1}B_2
+B_2A_{-1}
\bigr)v^*G_{0,0}
\\&\quad
-G_{0,0}v\bigl(
A_0
+A_{-1}B_1
+B_1A_{-1}
\bigr)v^*
G_1\\
&\quad-G_{0,1}v\bigl(
A_0
+A_{-1}B_1
+B_1A_{-1}
\bigr)v^*G_{0,0}
\\&\quad
-G_{0,0}vA_{-1}v^*G_{0,2}
-G_{0,1}vA_{-1}v^*G_{0,1}\\
&\quad-G_{0,2}vA_{-1}v^*G_{0,0}.
\end{align*}
Let us now use some vanishing relations coming from  
\eqref{am0b}, \eqref{16082220},
and \eqref{1608252242}:
\begin{align*}
G_0
&=
G_{0,0}
-G_{0,0}v\bigl(B_0
+A_1
+B_1A_{-1}B_1\\
&\qquad+A_{-1}B_2
+B_2A_{-1}
\bigr)v^*G_{0,0}
\\&\quad
-G_{0,0}vA_{-1}B_1v^*G_{0,1}
-G_{0,1}vB_1A_{-1}v^*G_{0,0}
\\&\quad
-G_{0,0}vA_{-1}v^*G_{0,2}
-G_{0,2}vA_{-1}v^*G_{0,0},
\end{align*}
and then 
insert expressions for $A_*$ and $B_*$,
noting the kernels of operators
and implementing \eqref{am0b} and \eqref{16082220}.
We omit some computations, obtaining
\begin{align*}
G_0
&=
G_{0,0}
-G_{0,0}v\Bigl(M_0^\dagger +Q
-m_1^\dagger m_3m_1^\dagger
\\&\qquad
-m_1^\dagger  M_2(M_0^\dagger +Q)
-(M_0^\dagger +Q) M_2 m_1^\dagger 
\Bigr)v^*G_{0,0}
\\&\quad
-G_{0,0}vm_1^\dagger v^*G_{0,2}
-G_{0,2}vm_1^\dagger v^*G_{0,0}
.
\end{align*}
Next we unfold $m_3$.
We use the expressions 
$m_3=QM_4Q-QM_2M_0^\dagger M_2Q-m_1m_1$
and $QM_2Q=m_1$
which hold under \eqref{16082220}, and then
\begin{align*}
G_0
&
=
G_{0,0}
-G_{0,0}v(I-m_1^\dagger  M_2)
M_0^\dagger 
(I- M_2m_1^\dagger)
v^*G_{0,0}
\\&\quad
-G_{0,0}v\bigl(
Q
+m_1^\dagger m_1m_1m_1^\dagger\\
&\qquad-m_1^\dagger  m_1
-m_1 m_1^\dagger 
\bigr)v^*G_{0,0}
\\&\quad
-G_{0,0}vm_1^\dagger v^*G_{0,2}
-G_{0,2}vm_1^\dagger v^*G_{0,0}\\
&\quad+G_{0,0}vm_1^\dagger M_4m_1^\dagger v^*G_{0,0}
.
\end{align*}
Now we note that by \eqref{160826} we have 
\begin{align}
m_1^{\dagger}
=-Uv^*P_0vU\label{16082619}
\end{align}
and this operator is bijective as $Q\mathcal K\to Q\mathcal K$.
Hence we have 
\begin{align*}
G_0
&
=
G_{0,0}
-(G_{0,0}+P_0VG_{0,2})
vM_0^\dagger v^*
(G_{0,0}+ G_{0,2}VP_0)
\\&\quad
+P_0VG_{0,2}
+G_{0,2}VP_0
+P_0V G_{0,4}VP_0
\end{align*}
Furthermore, we make use of the identities $VP_0=-H_0P_0$, 
$P_0V=-P_0H_0$ and $H_0G_{0,j}=G_{0,j}H_0=G_{0,j-2}$ for $j\ge 2$:
\begin{align*}
G_0
&
=
G_{0,0}
-(G_{0,0}-P_0G_{0,0})
vM_0^\dagger v^*
(G_{0,0}- G_{0,0}P_0)\\
&\quad
-P_0G_{0,0}
-G_{0,0}P_0
+P_0 G_{0,0}P_0
\\&
=(I-P_0)
\bigl[G_{0,0}\\
&\qquad
-G_{0,0}
v(U+v^*G_{0,0}v)^\dagger v^*G_{0,0}
\bigr](1-P_0)
.
\end{align*}
This verifies \eqref{ex2-G0}. 

The computation of $G_1$ in this case is very long,
and we do not present all the detail in this paper.
We only describe some of important steps. 
First we can write it, using only $A_*$ and $B_*$,
\begin{align*}
G_1
&
=
G_{0,1}
\\&\quad
-G_{0,0}vA_{-1}v^*G_{0,3}
-G_{0,1}vA_{-1}v^*G_{0,2}\\
&\quad
-G_{0,2}vA_{-1}v^*G_{0,1}
-G_{0,3}vA_{-1}v^*G_{0,0}
\\&\quad
-G_{0,0}v\bigl(
A_{-1}B_1
+B_1A_{-1}
+A_0
\bigr)v^*G_{0,2}
\\&\quad
-G_{0,1}v\bigl(
A_{-1}B_1
+B_1A_{-1}
+A_0
\bigr)v^*G_{0,1}
\\&\quad
-G_{0,2}v\bigl(
A_{-1}B_1+B_1A_{-1}
+A_0
\bigr)v^*G_{0,0}
\\&\quad
-G_{0,0}v\bigl(
B_0
+A_{-1}B_2
+B_1A_{-1}B_1\\
&\qquad+B_2A_{-1}
\bigr)v^*G_{0,1}
\\&\quad
-G_{0,1}v\bigl(
B_0
+A_{-1}B_2
+B_1A_{-1}B_1\\
&\qquad+B_2A_{-1}
\bigr)v^*G_{0,0}
\\&\quad
-G_{0,0}v\bigl(
B_1
+A_{-1}B_3
+B_1A_{-1}B_2\\
&\qquad+B_2A_{-1}B_1
+B_3A_{-1}+A_0B_2
+B_1A_0B_1
\\&
\qquad
+B_2A_0
+A_1B_1
+B_1A_1
+A_2
\bigr)v^*G_{0,0}.
\end{align*}
Then we insert the expressions of $A_*$ and $B_*$.
If we implement some of vanishing relations coming from 
\eqref{am0b}, \eqref{16082220}, and \eqref{1608252242},
we arrive at
\begin{align*}
G_1&
=
G_{0,1}
-G_{0,0}vm_1^\dagger v^*G_{0,3}
-G_{0,3}vm_1^\dagger v^*G_{0,0}
\\&\quad
-G_{0,0}v\bigl(
M_0^\dagger 
-m_1^\dagger M_2M_0^\dagger 
\bigr)v^*G_{0,1}\\
&\quad-G_{0,1}v\bigl(
M_0^\dagger 
-M_0^\dagger  M_2m_1^\dagger 
\bigr)v^*G_{0,0}
\\&\quad
-G_{0,0}v\Bigl[
-M_0^\dagger  M_1M_0^\dagger\\ 
&\qquad+m_1^\dagger (- M_3M_0^\dagger
+M_2M_0^\dagger  M_1M_0^\dagger 
)
\\&\qquad
+(-M_0^\dagger  M_3
+M_0^\dagger  M_1M_0^\dagger  M_2
)m_1^\dagger \\
&\qquad-m_1^\dagger m_4m_1^\dagger\Bigr]v^*G_{0,0}
.
\end{align*}
If we insert \eqref{16082619} and 
$m_4=QM_5Q-QM_2JM_3Q-QM_3JM_2Q$,
which holds especially in this case due to the vanishing relations
noted above, we come to 
\begin{align*}
G_1&
=
G_{0,1}
+G_{0,0}VP_0VG_{0,3}
+G_{0,3}VP_0VG_{0,0}
\\&\quad
-G_{0,0}\bigl(
vM_0^\dagger v^*
+VP_0VG_{0,2}v M_0^\dagger v^*
\bigr)G_{0,1}
\\&\quad
-G_{0,1}\bigl(
vM_0^\dagger v^*
+vM_0^\dagger  v^*G_{0,2}VP_0V
\bigr)G_{0,0}
\\&\quad
-G_{0,0}\Bigl[
-vM_0^\dagger  v^*G_{0,1}vM_0^\dagger v^*
+ VP_0VG_{0,3}vM_0^\dagger v^*
\\&\qquad
-VP_0VG_{0,2}v M_0^\dagger  v^*G_{0,1}vM_0^\dagger v^*\\
&\qquad+vM_0^\dagger  v^*G_{0,3}VP_0V
\\&\qquad
-vM_0^\dagger  v^*G_{0,1}vM_0^\dagger  v^*G_{0,2}VP_0V \\
&\qquad-VP_0VG_{0,5}VP_0V
\\&\qquad
+VP_0VG_{0,2}v M_0^\dagger v^*G_{0,3}VP_0V\\
&\qquad+VP_0VG_{0,3}vM_0^\dagger v^*G_{0,2}VP_0V
\Bigr]G_{0,0}
.
\end{align*}
Finally we use $VP_0=-H_0P_0$, $P_0V=-P_0H_0$
and \eqref{18082023},
and then the expression \eqref{ex2-G1} is obtained.
Hence we are done.
\end{proof}

\section{Exceptional threshold of the third kind}\label{1608241729}

Finally we prove Theorem~\ref{thm-ex3}.
Compared with the proof of Theorem~\ref{thm-ex2},
this case needs one more application of the inversion formula, or
Proposition~\ref{12.11.9.3.28},
and the formulas get much more complicated.

\begin{proof}[Proof of Theorem~\ref{thm-ex3}]
Let us repeat arguments of the previous section to some extent.
We write the expansion \eqref{M-expand} 
in the same form as \eqref{ex1-M-expand},
let $Q$ be the orthogonal projection onto $\Ker M_0$, and 
define $m(\kappa)$ by the same formula as \eqref{def-J0}.
Then by Proposition~\ref{12.11.9.3.28} we have the 
same formula as \eqref{16082255}, again. 
The operator $m(\kappa)$ defined by \eqref{def-J0}
has the same expansion as \eqref{16082317}
with \eqref{1608231743}--\eqref{1608231746},
but without \eqref{am0b} or \eqref{16082220}
by the assumption and Corollary~\ref{160616438}.
Now we apply the inversion formula, Proposition~\ref{12.11.9.3.28},
to the operator $m(\kappa)$.
Write the expansion  \eqref{16082317} in the form
\begin{equation}
m(\kappa)=m_0+\kappa \widetilde m_1(\kappa).
\label{ex1-M-expandbbb}
\end{equation}
The leading operator $m_0$ is non-zero and not invertible in $\mathcal B(Q\mathcal K)$ 
by the assumption and Corollary~\ref{160616438}. 
Let $T$ be the orthogonal projection onto $\mathop{\mathrm{Ker}}m_0\subset Q\mathcal K$, 
and set 
\begin{equation}\label{def-J0bbb}
q(\kappa)=
\sum_{j=0}^{\infty}(-1)^j
\kappa^jT\widetilde{m}_1(\kappa)\bigl[(m_0^\dagger +T)\widetilde{m}_1(\kappa)\bigr]^{j}T.
\end{equation}
Then we have by Proposition~\ref{12.11.9.3.28} that 
\begin{align}
m(\kappa)^\dagger &=(m(\kappa)+T)^\dagger\notag\\
&\quad +\frac{1}{\kappa}
(m(\kappa)+T)^\dagger  q(\kappa)^{\dagger}(m(\kappa)+T)^\dagger .
\label{16082255b}
\end{align}
Using \eqref{16082317} and 
\eqref{ex1-M-expandbbb}, let us write \eqref{def-J0bbb} in the form
\begin{equation*}
q(\kappa)=\sum_{j=0}^{\beta-4}\kappa^jq_j+\mathcal{O}(\kappa^{\beta-3});
\quad 
q_j\in\mathcal B(T\mathcal K).
\end{equation*}
The first and the second coefficients are given as 
\begin{align}\label{am0bb}
q_0&=Tm_1T,\quad
q_1=Tm_2T-Tm_1(m_0^\dagger +T)m_1T
.
\end{align}
Here we note that the leading operator $q_0$ has to be invertible
in $\mathcal B(T\mathcal K)$.
Otherwise, applying Proposition~\ref{12.11.9.3.28} once again,
we can show that $R(\kappa)$ has a singularity
of order $\kappa^{-j}$, $j\ge 3$ in its expansion.
This contradicts the self-adjointness of $H$.
Hence we can use the Neumann series 
to write $q(\kappa)^\dagger$, and obtain
\begin{equation}
q(\kappa)^\dagger=\sum_{j=0}^{\beta-4}\kappa^jA_j+\mathcal{O}(\kappa^{\beta-3}),\quad
A_j\in\mathcal B(T\mathcal K),
\label{ex1-m-expandbb}
\end{equation}
where 
\begin{align*}
&A_0=q_0^{\dagger},\quad
A_1=-q_0^\dagger q_1q_0^\dagger.
\end{align*}
We also write $(m(\kappa)+T)^\dagger$ 
employing the Neumann series as 
\begin{align}
(m(\kappa)+T)^\dagger
=\sum_{j=0}^{\beta-3}\kappa^jC_j+\mathcal{O}(\kappa^{\beta-2})
\label{1608225b}
\end{align}
with $C_j\in\mathcal B(Q\mathcal K)$ and
\begin{align*}
C_0=m_0^\dagger +T,\quad
C_1=-(m_0^\dagger +T) m_1(m_0^\dagger +T).
\end{align*}
We first insert the expansions 
\eqref{ex1-m-expandbb} and \eqref{1608225b} into \eqref{16082255b}:
\begin{align}
m(\kappa)^\dagger
&=\sum_{j=-1}^{\beta-5}\kappa^jD_j+\mathcal O(\kappa^{\beta-4}),
\label{160822553bb}\\
D_j&=C_j
+\sum_{\genfrac{}{}{0pt}{}{j_1\ge 0,j_2\ge 0,j_3\ge 0}{j_1+j_2+j_3=j+1}}
C_{j_1}A_{j_2}C_{j_3},
\notag
\end{align}
with $C_{-1}=0$.
Next, noting that we have an expansion of 
$(M(\kappa)+Q)^{-1}$ in the same form as \eqref{1608225},
we insert the expansions \eqref{160822553bb} and \eqref{1608225} 
into \eqref{16082255}:
\begin{align}
M(\kappa)^{-1}
&=\sum_{j=-2}^{\beta-6}\kappa^jE_j+\mathcal O(\kappa^{\beta-5}),
\label{160822553bbb}
\\
E_j&=B_j
+\sum_{\genfrac{}{}{0pt}{}{j_1\ge 0,j_2\ge -1,j_3\ge 0}{j_1+j_2+j_3=j+1}}
B_{j_1}D_{j_2}B_{j_3},
\notag
\end{align}
with $B_{-2}=B_{-1}=0$. We finally 
inserting the expansions \eqref{free-expan} with $N=\beta-4$ and \eqref{160822553bbb}
into \eqref{second-resolvent},
and then obtain the expansion
\begin{align*}
R(\kappa)&=\sum_{j=-2}^{\beta-6}\kappa^j
G_j+\mathcal O(\kappa^{\beta-5}),\\
G_j&=G_{0,j}
-\sum_{\genfrac{}{}{0pt}{}{j_1\ge 0,j_2\ge -2,j_3\ge 0}{j_1+j_2+j_3=j}}
G_{0,j_1}vE_{j_2}v^*G_{0,j_3},
\end{align*}
with $G_{0,-2}=G_{0,-1}=0$.

Next we compute the first two coefficients.
Let us start with $G_{-2}$.
Unfolding the above expressions, we can see with ease that 
\begin{align*}
G_{-2}
&=-G_{0,0}vE_{-2}v^*G_{0,0}\\
&=-G_{0,0}v\bigl(TM_2T-TM_1(M_0^\dagger +T)M_1T\bigr)^\dagger v^*G_{0,0}
.
\end{align*}
Since 
\begin{align}
m_0=QM_1Q=-\ket{Qv^*\mathbf{n}}\bra{Qv^*\mathbf{n}},
\label{16082323}
\end{align}
it follows that
\begin{align}
Tv^*\mathbf  n=TQv^*\mathbf n=0.
\label{16082320}
\end{align}
Hence we have 
\begin{align*}
G_{-2}
=-G_{0,0}v(Tv^*G_{0,2}vT)^\dagger v^*G_{0,0}
,
\end{align*}
and we can verify the identity $G_{-2}=P_0$ in exactly the same manner as in the 
proof of Theorem~\ref{thm-ex2}.

As for $G_{-1}$, it requires a slightly longer computations,
and we proceed step by step.
We can first write, concerning $A_*,B_*,C_*,D_*,E_*$ only, 
\begin{align*}
G_{-1}
&=
-G_{0,0}vE_{-1}v^*G_{0,0}
-G_{0,0}vE_{-2}v^*G_{0,1}\\
&\quad-G_{0,1}vE_{-2}v^*G_{0,0}
\\&=
-G_{0,0}v\Bigl(B_0\bigl(C_0+C_0A_1C_0\\
&\quad\qquad+C_0A_0C_1+C_1A_0C_0\bigr)B_0
\\&\qquad
+B_0C_0A_0C_0B_1+B_1C_0A_0C_0B_0\Bigr)v^*G_{0,0}
\\&\quad
-G_{0,0}vB_0C_0A_0C_0B_0v^*G_{0,1}\\
&\quad
-G_{0,1}vB_0C_0A_0C_0B_0v^*G_{0,0}.
\end{align*}
Next, we implement the identities $B_0C_*=C_*B_0=C_*$ and $C_0A_*=A_*C_0=A_*$,
insert expressions of $A_*,B_*,C_*$, 
and then use \eqref{16082320}:
\begin{align*}
G_{-1}
&=
-G_{0,0}v\Bigl(C_0+A_1+A_0C_1\\
&\qquad+C_1A_0
+A_0B_1+B_1A_0\Bigr)v^*G_{0,0}
\\&\quad
-G_{0,0}vA_0v^*G_{0,1}
-G_{0,1}vA_0v^*G_{0,0}
\\&
=
-G_{0,0}v\Bigl(m_0^\dagger +T-q_0^\dagger q_1q_0^\dagger
-q_0^\dagger  m_1(m_0^\dagger +T)\\
&\qquad-(m_0^\dagger +T) m_1q_0^\dagger
\\&\qquad
-q_0^\dagger M_1(M_0^\dagger +Q)
-(M_0^\dagger +Q) M_1q_0^\dagger \Bigr)v^*G_{0,0}
\\&\quad
-G_{0,0}vq_0^\dagger v^*G_{0,1}
-G_{0,1}vq_0^\dagger v^*G_{0,0}.
\\&
=
-G_{0,0}v\Bigl(m_0^\dagger +T-q_0^\dagger q_1q_0^\dagger
-q_0^\dagger  m_1(m_0^\dagger +T)\\
&\qquad-(m_0^\dagger +T) m_1q_0^\dagger
 \Bigr)v^*G_{0,0}.
\end{align*}
We further unfold $q_1$ and $m_1$ and use \eqref{16082320}:
\begin{align*}
G_{-1}
&=
-G_{0,0}v\Bigl(m_0^\dagger +T+q_0^\dagger M_2(m_0^\dagger +T)M_2q_0^\dagger
\\&\qquad
-q_0^\dagger  M_2(m_0^\dagger +T)
-(m_0^\dagger +T) M_2q_0^\dagger
 \Bigr)v^*G_{0,0}
\\&
=
-G_{0,0}v(I-q_0^\dagger M_2)m_0^\dagger (I-M_2q_0^\dagger)v^*G_{0,0}
\\&\quad
-G_{0,0}v(I-q_0^\dagger M_2)T(I-M_2q_0^\dagger)v^*G_{0,0}.
\end{align*}
Since $TM_2T=Tm_1T=q_0T$ by \eqref{16082320}, the last term can actually be removed:
\begin{align*}
G_{-1}
=
-G_{0,0}v(I-q_0^\dagger M_2)m_0^\dagger (I-M_2q_0^\dagger)v^*G_{0,0}.
\end{align*}
Finally by \eqref{16082323} we can write 
\begin{align*}
m_0^\dagger =-\|Qv^*\mathbf{n}\|^{-4}\ket{Qv^*\mathbf{n}}\bra{Qv^*\mathbf{n}},
\end{align*}
and hence we obtain 
\begin{align*}
G_{-1}&
=|\Psi_c\rangle\langle \Psi_c|,\\
\Psi_c&=\|Qv^*\mathbf{n}\|^{-2}G_{0,0}v(I-q_0^\dagger 
v^*G_{0,2}v)Qv^*\mathbf n
\in \mathcal E.
\end{align*}
Let us verify that the above $\Psi_c$ is in fact the canonical resonance function.
For any $\Psi\in\mathsf E$ set $\Phi=w\Psi\in T\mathcal K$.
As in the proof of Theorem~\ref{thm-ex2} we can verify that 
\begin{align*}
\langle\Psi,\Psi_c\rangle
&=-\|Qv^*\mathbf{n}\|^{-2}\\
&\quad\times
\bigl\langle G_{0,0}vT\Phi,
G_{0,0}v(I-q_0^\dagger v^*G_{0,2}v)Qv^*\mathbf n\bigr\rangle\\
&=0.
\end{align*}
We can also prove that 
\begin{align*}
\langle V\mathbf n,\Psi_c\rangle 
&=\|Qv^*\mathbf{n}\|^{-2}\\
&\quad\times\bigl\langle V\mathbf n,
G_{0,0}v(I-q_0^\dagger v^*G_{0,2}v)Qv^*\mathbf n\bigr\rangle
\\&
=\|Qv^*\mathbf{n}\|^{-2}\\
&\quad\times\bigl\langle Uv^*\mathbf n,
(M_0-U)(I-q_0^\dagger v^*G_{0,2}v)Qv^*\mathbf n\bigr\rangle
\\&
=-\|Qv^*\mathbf{n}\|^{-2}
\bigl\langle v^*\mathbf n,(I-q_0^\dagger v^*G_{0,2}v)Qv^*\mathbf n\bigr\rangle
\\&=-1.
\end{align*}
This concludes the proof.
\end{proof}

\appendix
\section{Inversion formula}\label{1608211941}

In this appendix we present an inversion formula
needed in the proof of the main results of the paper.
The formula is quoted from Ito-Jensen\cite[Section 3.1]{IJ}, which 
in turn was adapted from Jensen-Nenciu\cite[Corollary 2.2]{JN0}.

Let us argue in a general context.

\begin{assumption}\label{12.11.9.1.54}
Let ${\mathcal K}$ be a Hilbert space and 
$A(\kappa)$ a family of bounded operators on ${\mathcal K}$
with $\kappa\in D\subset \mathbb{C}\setminus \{0\}$.
Suppose that
\begin{enumerate}
\item
The set $D\subset \mathbb{C}\setminus \{0\}$ 
is invariant under complex conjugation and accumulates at $0\in \mathbb{C}$.
\item\label{12.12.19.2.58}
For each $\kappa\in D$ the operator $A(\kappa)$ satisfies 
$A(\kappa)^*=A(\overline{\kappa})$ and has a bounded inverse $A(\kappa)^{-1}
\in \mathcal B(\mathcal K)$.
\item
As $\kappa\to 0$ in $D$, the operator $A(\kappa)$ has an expansion
in the uniform topology of the operators at ${\mathcal K}$:
\begin{align}
A(\kappa)=A_0+\kappa \widetilde A_1(\kappa);\quad  \widetilde A_1(\kappa)
={\mathcal O}(1).\label{12.11.9.3.40}
\end{align}
\item
The spectrum of $A_0$ 
does not accumulate at $0\in\mathbb{C}$ as a set.
\end{enumerate}
\end{assumption}

If the leading operator $A_0$ is invertible in $\mathcal B(\mathcal K)$, the Neumann series 
provides an inversion formula for the expansion of $A(\kappa)^{-1}$:
\begin{align*}
A(\kappa)^{-1}=\sum_{j=0}^\infty (-1)^j\kappa^j
A_0^{-1}\bigl[\widetilde A_1(\kappa)A_0^{-1}\bigr]^j.
\end{align*}
The inversion formula given below is useful
when $A_0$ is not invertible in $\mathcal B(\mathcal K)$.

We define the \textit{pseudo-inverse} $a^{\dagger}$ for 
a complex number $a\in \mathbb C$ by
\begin{align}
a^\dagger=
\begin{cases}
0&\text{if $a=0$},\\
a^{-1} & \text{if $a\neq 0$}.
\end{cases}
\label{13.3.5.22.28}
\end{align} 
Let $\mathcal K'\subset \mathcal K$ be a closed subspace.
We always identify $\mathcal B(\mathcal K')$
with its embedding in $\mathcal B(\mathcal K)$ in the standard way.
For an operator $A\in \mathcal B(\mathcal K')\subset 
\mathcal B(\mathcal K)$ 
we say that $A$ is \textit{invertible in $\mathcal B(\mathcal K')$} 
if there exists an operator $A^\dagger \in \mathcal B(\mathcal K')$
such that $A^\dagger A=AA^\dagger=I_{\mathcal K'}$,
which we identify with the orthogonal projection onto $\mathcal K'\subset \mathcal K$ as noted.
For a general self-adjoint operator $A$ on $\mathcal K$
we abuse the notation $A^\dagger$ also to denote the operator 
defined by the usual operational calculus 
for the function (\ref{13.3.5.22.28}).
The operator $A^\dagger$ for a self-adjoint operator $A$
belongs to $B(\mathcal K)$ 
if and only if the spectrum of $A$ 
does not accumulate at $0$ as a set,
and in such a case the above two $A^\dagger$ coincide.
In either case we call $A^\dagger$ the \textit{pseudo-inverse} of $A$.
The reader should note that we always use the notation $A^*$ for the adjoint  and the notation $A^\dagger$ for the pseudo-inverse.

\begin{proposition}\label{12.11.9.3.28}
Suppose Assumption~\ref{12.11.9.1.54}.
Let $Q$ be the orthogonal projection onto $\mathop{\mathrm{Ker}} A_0$,
and define the operator $a(\kappa)\in\mathcal B(Q{\mathcal K})$
by
\begin{align}
\begin{split}
a(\kappa)&=\tfrac{1}{\kappa}
\bigl\{I_{Q\mathcal K}-Q(A(\kappa)+Q)^{-1}Q\bigr\}\\
&=\sum_{j=0}^\infty (-1)^j\kappa^jQ\widetilde A_1(\kappa)
\bigl[(A_0^\dagger+Q)\widetilde A_1(\kappa)\bigr]^{j}Q.
\end{split}
\label{12.11.9.3.30}
\end{align}
Then $a(\kappa)$ is bounded in ${\mathcal B}(Q{\mathcal K})$ 
as $\kappa\to 0$ in $D$.
Moreover, for each $\kappa\in D$ sufficiently close to $0$ 
the operator $a(\kappa)$ is invertible in $\mathcal B(Q{\mathcal K})$, and 
\begin{align}
A(\kappa)^{-1}
&=(A(\kappa)+Q)^{-1}\notag\\
&\quad+\frac{1}{\kappa}(A(\kappa)+Q)^{-1}a(\kappa)^\dagger(A(\kappa)+Q)^{-1}.
\label{12.11.9.5.30}
\end{align}
\end{proposition}

\bigskip

\section*{acknowledgements}
The authors would like to thank Shu Nakamura 
for commenting on a general boundary condition.
KI was partially supported by JSPS KAKENHI Grant Number JP25800073.
The authors were partially supported by the Danish Council 
for Independent Research $|$ Natural Sciences, Grant DFF--4181-00042.

\end{document}